\documentclass[12pt]{article}

\usepackage{acro}
\usepackage{algorithm}
\usepackage{algpseudocode}
\usepackage{amsmath}
\usepackage{amssymb}
\usepackage{amsthm}
\usepackage{bbm}
\usepackage{bm}
\usepackage{centernot}
\usepackage{color}
\usepackage{fullpage}
\usepackage{graphicx}
\usepackage{hyperref}
\usepackage{natbib}
\usepackage{subcaption}
\usepackage{cleveref}
\usepackage{xcolor}

\newtheorem{theorem}{Theorem}

\newtheorem{proposition}{Proposition}
\newtheorem{remark}{Remark}
\newtheorem{corollary}{Corollary}

\DeclareAcronym{PC}{short = PrO, long = predictively oriented}
\DeclareAcronym{SVGD}{short = SVGD, long = Stein variational gradient descent}
\DeclareAcronym{VGD}{short = VGD, long = variational gradient descent}
\DeclareAcronym{MCMC}{short = MCMC, long = Markov chain Monte Carlo}
\DeclareAcronym{MFLD}{short = MFLD, long = mean field Langevin dynamics}
\DeclareAcronym{RKHS}{short = RKHS, long = reproducing kernel Hilbert space}
\DeclareAcronym{KLD}{short = KLD, long = Kullback--Leibler divergence}
\DeclareAcronym{KGD}{short = KGD, long = kernel gradient discrepancy}
\DeclareAcronym{KSD}{short = KSD, long = kernel Stein discrepancy}
\DeclareAcronym{ODE}{short = ODE, long = ordinary differential equation}
\DeclareAcronym{MMD}{short = MMD, long = maximum mean discrepancy}

\DeclareMathOperator*{\argmin}{arg\,min}

\newcommand{\nll}{\centernot{\ll}}
\newcommand{\vargrad}{\nabla_{\mathrm{V}}}
\newcommand{\QBayes}{Q_{\mathrm{Bayes}}}
\newcommand{\QPC}{Q_{\mathrm{PrO}}}

\newcommand{\PBayes}{P_{\mathrm{Bayes}}}
\newcommand{\PPC}{P_{\mathrm{PrO}}}
\newcommand{\LBayes}{\mathcal{L}_{\mathrm{Bayes}}}
\newcommand{\LPC}{\mathcal{L}_{\mathrm{PrO}}}
\newcommand{\KLD}{\mathrm{KLD}}

\newif\ifshowalttext
\showalttextfalse
\newcommand{\alttext}[1]{%
\ifshowalttext
{\color{green!50!black}Alt text: #1}%
\fi
}

\usepackage{changepage}   %

\begin{document}

\title{Detecting Model Misspecification in Bayesian Inverse Problems via Variational Gradient Descent}
\author{Qingyang Liu$^1$, Matthew A. Fisher$^1$, Zheyang Shen$^1$, \\
Xuebin Zhao$^2$, Katherine Tant$^3$, Andrew Curtis$^2$, Chris. J. Oates$^1$\footnote{Correspondence should be sent to \texttt{chris.oates@ncl.ac.uk}.} \\
\small $^1$Newcastle University, UK \\
\small $^2$University of Edinburgh, UK \\
\small $^3$University of Glasgow, UK}
\maketitle

\begin{abstract}
    Bayesian inference is optimal when the statistical model is well-specified, while outside this setting Bayesian inference can catastrophically fail; accordingly a wealth of post-Bayesian methodologies have been proposed.
Predictively oriented (PrO) approaches lift the statistical model $P_\theta$ to an (infinite) mixture model $\int P_\theta \; \mathrm{d}Q(\theta)$ and fit this predictive distribution via minimising an entropy-regularised objective functional.
In the well-specified setting one expects the mixing distribution $Q$ to concentrate around the true data-generating parameter in the large data limit, while such singular concentration will typically not be observed if the model is misspecified.
Our contribution is to demonstrate that one can empirically detect model misspecification by comparing the standard Bayesian posterior to the PrO `posterior' $Q$, and that each can be efficiently computed via variational gradient descent.
A simulation study confirms that model misspecification can be rigorously tested using this framework, and a detailed case study in seismic tomography illustrates the potential of this approach as a novel and widely-applicable diagnostic tool for the standard Bayesian workflow.
\end{abstract}

\section{Introduction}
\label{sec: intro}

Detecting and mitigating model misspecification is a pertinent but difficult issue in Bayesian statistics, where one must consider the full posterior predictive distribution \emph{in lieu} of e.g. plotting a simple scalar residual.
The two main approaches to detecting model misspecification (also called \emph{model criticism}) are: 
\begin{enumerate}
\item \emph{Predictive}:  Compare the observed data to the posterior predictive distribution, either re-using the data that were used to construct it, as in posterior predictive checks \citep{rubin1984bayesianly,gelman1996posterior,bayarri2000p}, or using held-out data \citep{moran2024holdout}.
\item \emph{Comparative}:  Perform model selection over a set of candidates models in which the current model is contained \citep[e.g.][]{kass1995bayes,wasserman2000bayesian,kamary2014testing}.
\end{enumerate}
In the first case, if the data are in some sense `unexpected' under the posterior predictive distribution, this serves as evidence that either the prior or the model may be misspecified.
When the same data are used both to construct the predictive distribution and to perform the check, the resulting posterior predictive $p$-values are known to be conservative; remedies include calibration by simulation \citep{hjort2006post}, modified $p$-values that remove the information in the test statistic from the posterior \citep{bayarri2000p}, and the use of held-out data \citep{moran2024holdout}.
In the second case, if the data strongly support an alternative model, this suggests that the original model may be misspecified.
Of course, this is not a true dichotomy and the predictive and comparative approaches are intimately related; for instance, predictive performance is a common criteria for selecting a suitable model \citep{piironen2017comparison,fong2020marginal}.

For challenging applications, both of these approaches can be impractical.
An example of such a challenging application is \emph{seismic travel time tomography} \citep{curtis2002probing, zhang2020seismic, zhao2022bayesian}, a regression task where one seeks to reconstruct a subsurface seismic velocity field using measured first arrival times of seismic waves travelling between seismic source and sensor (seismometer) locations.
Evaluation of the likelihood and/or its gradient here is associated with a nontrivial computational cost, since  the \emph{Eikonal equation}, a partial differential equation describing the high frequency approximation of the wave equation, needs to be solved to estimate the travel times of the first arriving seismic waves.
It is worth pointing out that `misspecification' as discussed in this work refers to the suitability of the statistical model for the dataset, rather than any errors that may be present in the physical model for seismic wave travel, though the two issues are of course closely related.

Considering the predictive approach in the seismic tomography context, calibrating a check that re-uses the data requires the posterior to be recomputed for many synthetic datasets, each requiring repeated solution of the Eikonal equation, while for checks based on held-out data the spatiotemporal nature of the observations renders the data strongly dependent, representing a challenge in constructing a suitable held-out dataset.
Further, while in principle there are solutions to cross-validation for dependent data \citep[e.g.][]{burman1994cross,rabinowicz2022cross}, the computational cost associated with performing multiple folds of held-out prediction can render this approach impractical.

Focussing instead on the comparative approach, constructing plausible alternative models requires modifying the physical assumptions underlying the seismic wave propagation.  This in turn requires the implementation and optimisation of suitable numerical methods, which demands considerable effort.
In the case of model misspecification, there is particular interest in exploring \emph{more sophisticated} models, but without advanced physical insight, constructing such alternatives might be impractical.

The aim of this paper is to propose a \emph{simple, practical and general approach to detecting model misspecification in Bayesian statistics}, and we focus on the seismic travel time tomography problem to demonstrate its potential.
Our solution can be considered both a predictive and a comparative approach; predictive because we explicitly assess the predictive performance of the model, and comparative because our principal innovation is to automate the generation of a candidate model set.

Inspired by \emph{nonparametric maximum likelihood} \citep{laird1978nonparametric}, we consider lifting a statistical model $P_\theta$ to an (infinite) mixture model $P_Q := \int P_\theta \; \mathrm{d}Q(\theta)$ as a mechanism to generate an infinite candidate set of alternative models (parametrised by $Q$), while retaining any causal semantics present in the original statistical model.
However, the nonparametric maximum likelihood estimate for $Q$ is not computable in general.
Instead, we therefore follow the \emph{predictively oriented} (\acs{PC}) \emph{posterior} approach of \citet{lai2024predictive,shen2025prediction,mclatchie2025predictively}, in which learning of the mixing density $Q$ proceeds based on the performance of the predictive distribution $P_Q$ (in this sense our method is a predictive approach) and is cast as an entropy-regularised optimisation task, with the solution being a \ac{PC} `posterior' denoted $\QPC$.
Our specific contributions are as follows:
\begin{itemize}
    \item \emph{Detecting misspecification:}  We propose a direct comparison of the predictive distributions associated with $\QPC$ and the Bayesian posterior $\QBayes$ as a general strategy enabling model misspecification to be detected (in this sense our method is a comparative approach).
    Indeed, \citet{mclatchie2025predictively} proved that in the well-specified setting the predictive distribution associated to $\QPC$ converges to the true data-generating distribution in the large data limit, in agreement with the predictive distribution associated to $\QBayes$, while such agreement is typically not observed when the model is misspecified.
    \item \emph{Testing for misspecification:} A formal hypothesis test for model misspecification is presented, and {for regression models that are linear in the parameters we establish sufficient conditions for the resulting test to have asymptotically correct Type~I error.}
    Due to the parametric bootstrap, our test is practical only for statistical models $P_\theta$ for which both simulation and inference can be rapidly performed, motivating the development of efficient numerical methods to obtain $\QPC$ and $\QBayes$.
    \item \emph{Computation as variational gradient descent:}  Since $\QPC$ is defined as the minimiser of a \emph{nonlinear} variational objective, we do not have access to an unnormalised form of the target, and as such standard methods such as \ac{MCMC} cannot be immediately applied.
    As a solution, we turn to \emph{variational gradient descent} \citep[\acs{VGD};][]{wang2019nonlinear,chazal2025computable}, which is a nonlinear generalisation of \ac{SVGD} \citep{liu2016stein}, a popular numerical method for computing $\QBayes$.
    Novel sufficient conditions for the consistency of \ac{VGD} are presented, than can be verified in the seismology context.
    Remarkably, sampling from $\QPC$ can be achieved with a one-line change to standard \ac{SVGD}, enabling both $\QPC$ and $\QBayes$ to be computed using identical code, with an additional argument specifying whether the \ac{PC} or Bayesian posterior is to be computed.
    Using \ac{VGD}, we empirically confirm the ability of the parametric bootstrap hypothesis test to detect when the statistical model is misspecified.
    \item \emph{Application to inverse problems:}  For challenging settings where testing for misspecification using a parametric bootstrap would be computationally impractical, even using \ac{VGD}, we investigate whether a visual comparison of $\QPC$ and $\QBayes$ can still act as a useful diagnostic tool.
    To this end a detailed case study involving seismic travel time tomography is presented.
    Here we indeed find that a visual comparison of $\QPC$ and $\QBayes$ is able to distinguish between the well-specified case and scenarios in which the location of the sensors is misspecified.
\end{itemize}

\noindent Our methods are contained in \Cref{sec: methods}, while our empirical assessment is contained in \Cref{sec: empirical}.
A summary of our findings is presented in \Cref{sec: discuss}, with our proofs and experimental protocol reserved for the Appendix.

\section{Methods}
\label{sec: methods}

The variational formulation of Bayesian updating and the related \ac{PC} approach are recalled in \Cref{subsec: set-up}.
The \ac{VGD} methodology is presented in \Cref{sec: var grads}, and novel theoretical analysis required to establish its validity in our context is presented in \Cref{sec: theory}.
Our formal hypothesis test for misspecification is presented in \Cref{sec: implement} and {sufficient conditions for the validity of the parametric bootstrap are presented} in \Cref{subsec: bootstrap}.

\subsection{Bayesian and Predictively Oriented Approaches}
\label{subsec: set-up}

Let $P_\theta$ denote a statistical model parametrised by $\theta \in \mathbb{R}^d$, whose density $p_\theta$ we assume to exist.
Let $\mathfrak{D}_n$ denote the dataset.
In what follows, motivated by our seismic tomography case study, we focus on regression modelling, where responses $\{y_i\}_{i=1}^n$ are conditionally independent given covariates $\{x_i\}_{i=1}^n$, so that
\begin{align*}
    \log p_\theta(\mathfrak{D}_n) = \sum_{i=1}^n \log p_\theta(y_i | x_i)
\end{align*}
where the dependence on the covariates $x_i$ is made explicit.
However it should be noted that our methods are applicable beyond the regression context.

\paragraph{Standard Bayesian Posterior}

Let $\mathcal{P}(\mathbb{R}^d)$ denote the set of distributions\footnote{Measurability is implicitly assumed in this manuscript.} on $\mathbb{R}^d$.
Let $Q \ll Q_0$ denote that $Q$ is absolutely continuous with respect to $Q_0$, and $\mathrm{d} Q / \mathrm{d} Q_0$ the Radon--Nikodym density of $Q$ with respect to $Q_0$.
For $Q \ll Q_0$, the \ac{KLD} is defined as $\KLD(Q \Vert Q_0) := \int \log (\mathrm{d} Q / \mathrm{d} Q_0)\; \mathrm{d} Q$, while for $Q \nll Q_0$ we set $\KLD(Q || Q_0) = \infty$. 
Recall the variational characterisation of the standard Bayesian posterior due to \citet{zellner1988optimal}:
\begin{align}
    \QBayes := \argmin_{Q \in \mathcal{P}(\mathbb{R}^d)} \;  - \sum_{i=1}^n \int \log p_\theta(y_i | x_i) \; \mathrm{d}Q(\theta) + \KLD(Q || Q_0) \label{eq: first QBayes}
\end{align}
where $Q_0 \in \mathcal{P}(\mathbb{R}^d)$ is the prior distribution \citep[see also e.g.][]{knoblauch2022optimization}.
For the integral to be well-defined, i.e. for $\theta \mapsto - \log p_\theta(\mathfrak{D}_n)$ to be $Q$-integrable for all $Q \in \mathcal{P}(\mathbb{R}^d)$, it is sufficient for $\theta \mapsto p_\theta(\mathfrak{D}_n)$ to be bounded.

\paragraph{Predictively Oriented Posterior}

The \ac{PC} approaches of \citet{lai2024predictive,shen2025prediction,mclatchie2025predictively} were developed with the aim of avoiding over-confident predictions when the statistical model is misspecified.
These approaches lift the original parametric model $P_\theta$ into a mixture model $P_Q$, with density
$$
p_Q(y_i | x_i) = \int p_\theta(y_i | x_i) \; \mathrm{d}Q(\theta) ,
$$
and then attempt to learn $Q$ by minimising an entropy-regularised objective functional.
For the purpose of this paper we measure the suitability of $Q$ using the (relative) entropy-regularised mixture log-likelihood
\begin{align}
    \QPC := \argmin_{Q \in \mathcal{P}(\mathbb{R}^d)} \; - \sum_{i=1}^n \log p_Q(y_i | x_i) + \KLD(Q || Q_0)  . \label{eq: def QPC}
\end{align}
This can be viewed as an entropy-regularised form of \emph{nonparametric maximum likelihood} \citep{laird1978nonparametric}; the entropic regularisation is a key ingredient, since otherwise the solution will be atomic \citep[][e.g. Theorem 21 in Chapter 5]{lindsay1995mixture}.
For a discussion of other related work, such as \citet{masegosa2020learning,sheth2020pseudo,jankowiak2020deep,jankowiak2020parametric,morningstar2022pacm}, see \citet{shen2025prediction,mclatchie2025predictively}.

Since the variational formulation \eqref{eq: def QPC} is non-standard, we should first ask if $\QPC$ is well-defined.
Let $\mathcal{P}_\alpha(\mathbb{R}^d)$ denote the subset of $\mathcal{P}(\mathbb{R}^d)$ for which moments of order $\alpha$ exist.
The proof of the following result is contained in \Cref{app: well defined}:

\begin{theorem}[$\QPC$ is well-defined] \label{prop: Qmix well def}
    Let $Q_0$ admit a positive density $q_0$ on $\mathbb{R}^d$.
    Let $\theta \mapsto p_\theta(y_i | x_i)$ be positive, bounded and continuous for each $(x_i,y_i)$ in the dataset.
    Then there exists a unique solution to \eqref{eq: def QPC}.
    Further, if $Q_0 \in \mathcal{P}_\alpha(\mathbb{R}^d)$ then $\QPC \in \mathcal{P}_\alpha(\mathbb{R}^d)$.
\end{theorem}

The benefit of lifting to a mixture model is as follows:  
It the original statistical model $P_\theta$ was well-specified, so that there really was a correct parameter $\theta_\star$, then we can hope $\QPC$ concentrates around $\theta_\star$ (i.e. collapses to a mixture model with a single mixture component).
Likewise one would expect vanishing posterior uncertainty in the standard Bayesian context.
Inspecting whether the learned $\QPC$ appears to converge to the same limit as the standard Bayesian posterior $\QBayes$ can therefore provide a useful validation that the model is well-specified.
On the other hand, if the original statistical model was misspecified, then we expect $\QPC$ to learn a non-trivial mixture model, assuming such a mixture provides a better explanation of the dataset than any single instance of $P_\theta$ could.
That is, $\QPC$ is able to \emph{adapt} to the level of model misspecification, in a way that standard Bayesian inference cannot.
These intuitions for the asymptotic behaviour of $\QPC$ are confirmed in the recent detailed theoretical treatment in \citet{mclatchie2025predictively}.
An empirical demonstration of the effectiveness of this approach is the subject of \Cref{sec: empirical}; the remainder of this section addresses the key practical question of how to calculate $\QPC$ in \eqref{eq: def QPC}.

\begin{remark}[Comparison to mixture models]
\label{rem: mixture}
    One can always ask whether a mixture model provides a better explanation of the data compared to any single instance of the original statistical model.
    The \ac{PC} posterior approach is fundamentally different to fitting a mixture model; there is no prior on the number of mixture components, and one does not need to extend the dimension of the parameter space as would ordinarily happen when mixture models are considered.
\end{remark}

\begin{remark}[Choice of learning rate]
\label{rem: learning rate}
Note that, unlike generalised Bayesian methods \citep{bissiri2016general,knoblauch2022optimization} and in contrast to the earlier work on \ac{PC} approaches in \citet{lai2024predictive,shen2025prediction,mclatchie2025predictively}, no learning rate appears in \eqref{eq: def QPC}.
Earlier works introduced a learning rate $\lambda$ in the form of $\lambda \times \KLD(Q || Q_0)$ to accommodate other choices of data-dependent loss, such as maximum mean discrepancy.
Since our aim is to assess the adequacy of the standard Bayesian analysis, $\lambda = 1$ is arguably a natural default, matching how \ac{KLD} appears in the definition of the standard Bayesian posterior $\QBayes$ in \eqref{eq: first QBayes}, and for consistency we adopt the same value for $\QPC$.
In fact, the methods that we develop in this work are remarkably insensitive to the choice of $\lambda$; see \Cref{rem: lr sensitivity}.
Selection of learning rates is known to be difficult \citep{wu2023comparison} and it is therefore advantageous that this default does not need to be tuned.
\end{remark}

\subsection{Variational Gradient Descent}
\label{sec: var grads}

An immediate question is \emph{how to solve \eqref{eq: def QPC}}?
Since the parameter of the mixture model is $Q$, it is unclear how to proceed; $Q$ lives in $\mathcal{P}(\mathbb{R}^d)$ which is not a vector space, making it unclear how to apply operations such as taking a gradient with respect to $Q$.
To resolve this problem we consider a general entropy-regularised variational objective 
\begin{align}
    \mathcal{J}(Q) := \mathcal{L}(Q) + \KLD(Q || Q_0) , \label{eq: objective}
\end{align}
which accommodates both $\QBayes$ and $\QPC$ by taking the loss function to be, respectively, either
\begin{align}
\LBayes(Q) & = - \sum_{i=1}^n \int \log p_\theta(y_i | x_i) \; \mathrm{d}Q(\theta), \label{eq: define losses Bayes} \end{align}
for the standard Bayesian posterior, or
\begin{align}
    \LPC(Q) = - \sum_{i=1}^n \log \int p_\theta(y_i | x_i) \; \mathrm{d}Q(\theta) ,  \label{eq: define losses}
\end{align}
for the \ac{PC} posterior, which differ only in the order in which the integral and the logarithm are performed.
Our aim is a rigorous notion of gradient descent that can be applied to (relative) entropy-regularised objective in \eqref{eq: objective}.
To this end we now explain how the \ac{VGD} algorithm, due originally to \citet{wang2019nonlinear}, can be applied in our context.

\paragraph{Variational Gradient}

The notion of a gradient that we will need is a \emph{variational gradient}.
For a suitably regular functional $\mathcal{F} : \mathcal{P}(\mathbb{R}^d){} \rightarrow \mathbb{R}$, the \emph{first variation} at $Q \in \mathcal{P}(\mathbb{R}^d)$ is defined as a map $\mathcal{F}'(Q) : \mathbb{R}^d \rightarrow \mathbb{R}$ such that 
$$
\lim_{\epsilon\to 0}\frac{1}{\epsilon} \{ \mathcal{F}(Q+\epsilon \chi)-\mathcal{F}(Q) \} = \int\mathcal{F}'(Q) \; \mathrm{d}\chi
$$ 
for all perturbations $\chi$ of the form $\chi=Q' - Q$ with $Q' \in \mathcal{P}(\mathbb{R}^d)$; note that if it exists, the first variation is unique up to an additive constant. 
The first variation extends the standard derivative concept to arguments that are distribution-valued.
Much of the intuition for standard derivatives carries over to the first variation; for example, a linear functional $\mathcal{F}_1(Q) = \int f(\theta) \, \mathrm{d}Q(\theta)$ has a first variation $\mathcal{F}_1'(Q)(\theta) = f(\theta)$ which is constant in $Q$.
Similarly a quadratic functional $\mathcal{F}_2(Q) = \iint f(\theta,\vartheta) \, \mathrm{d}Q(\theta) \mathrm{d}Q(\vartheta)$, with $f$ symmetric, has a first variation $\mathcal{F}_2'(Q)(\theta) = 2 \int f(\theta , \vartheta) \, \mathrm{d}Q(\vartheta)$ which is linear in $Q$.
(In both of these examples $f$ must be regular enough for the first variation to be well-defined.)

Given a functional $\mathcal{F}(Q)$ we define the variational gradient of $\mathcal{F}$ at $Q$ as the function $\vargrad \mathcal{F}(Q)(\theta) := \nabla_\theta \mathcal{F}'(Q)(\theta)$ where $\mathcal{F}'(Q)$ is the first variation of $\mathcal{F}$ at $Q$ \citep[][Definition 1]{chazal2025computable}.
The variational gradient generalises the better-known \emph{Wasserstein gradient} from optimal transport beyond the Wasserstein space $Q \in \mathcal{P}_2(\mathbb{R}^d)$, so that we do not need to assume a second moment; see e.g. \citet{otto2001geometry}. 
For the examples in the previous paragraph, again with $f$ sufficiently regular, the variational gradients are $\nabla_{\mathrm{V}} \mathcal{F}_1(Q)(\theta) = \nabla_\theta f(\theta)$ and $\nabla_{\mathrm{V}} \mathcal{F}_2(Q)(\theta) = 2 \int \nabla_\theta f(\theta,\vartheta) \, \mathrm{d}Q(\vartheta)$.
For the loss functions in \eqref{eq: define losses Bayes} and \eqref{eq: define losses} we have
\begin{align}
\vargrad \LBayes(Q)(\theta) & = - \sum_{i=1}^n \nabla_\theta \log p_\theta(y_i | x_i) \label{eq: vargrad LBayes} \\
\vargrad \LPC(Q)(\theta) & = - \sum_{i=1}^n w_\theta^Q(y_i | x_i) \nabla_\theta \log p_\theta(y_i | x_i) , \qquad w_\theta^Q(y_i | x_i) := \frac{ p_\theta(y_i | x_i) }{ p_Q(y_i | x_i) } ,   \label{eq: vargrad LPC}
\end{align}
see \Cref{prop: explicit variational grad} in \Cref{app: prelim}.
Note that \eqref{eq: vargrad LPC} can be seen as a weighted version of \eqref{eq: vargrad LBayes}, which agrees when $Q = \delta_\theta$ is a Dirac distribution at $\theta \in \mathbb{R}^d$.

\paragraph{Computing Directional Derivatives}

Let $T_\# Q$ denote the distribution of $T(X)$ where $X \sim Q$.
Consider the directional derivatives 
$$
\left. \frac{\mathrm{d}}{\mathrm{d}\epsilon} \mathcal{J}((\mathrm{I}_d + \epsilon v)_\# Q) \right|_{\epsilon = 0} 
$$
as specified by a suitable vector field $v : \mathbb{R}^d \rightarrow \mathbb{R}^d$, where $\mathrm{I}_d$ is the identity map on $\mathbb{R}^d$.
For the purpose of optimisation, we seek a vector field $v$ for which the rate of decrease in $\mathcal{J}$ is maximised. 
To this end, letting
\begin{align*}
	 \mathcal{T}_{Q} v(\theta) & := \left[ (\nabla\log q_0)(\theta) - \vargrad  \mathcal{L}(Q)(\theta)  \right] \cdot v(\theta) + (\nabla \cdot v)(\theta) ,
\end{align*}
it follows from the fundamental theorem of calculus \citep[see e.g. Section 3.2.2.2 of][]{chazal2025computable} that
\begin{align}
    \left. \frac{\mathrm{d}}{\mathrm{d}\epsilon} \mathcal{J}((\mathrm{I}_d + \epsilon v)_\# Q) \right|_{\epsilon = 0} = - \int \mathcal{T}_{Q} v(\theta) \; \mathrm{d} Q(\theta) .  \label{eq: KL to KGD}
\end{align}
That is, the directional derivative of the objective $\mathcal{J}$ in \eqref{eq: objective} can be expressed as an explicit $Q$-dependent linear functional applied to the vector field.

\paragraph{Following the Directions of Steepest Descent}

Next we pick the vector field $v_Q$ from the unit ball of an appropriate Hilbert space for which the magnitude of the negative gradient in \eqref{eq: KL to KGD} is maximised.
For a multivariate function, let $\partial_{i,j}$, $\nabla_i$, etc, indicate the action of the differential operators with respect to the $i$th argument.
Letting $\mathcal{H}_k$ denote the \ac{RKHS} associated to a symmetric positive semi-definite kernel $k : \mathbb{R}^d \times \mathbb{R}^d \rightarrow \mathbb{R}$, we seek $v_Q \in \mathcal{H}_k^d$, the $d$-fold Cartesian product, which leads to
\begin{align}
    v_Q(\cdot) \propto \int \{k(\theta,\cdot ) (\nabla \log q_0 - \vargrad  \mathcal{L}(Q))(\theta) + \nabla_1 k(\theta, \cdot) \} \; \mathrm{d} Q(\theta).  \label{eq: steepest}
\end{align}
To numerically approximate this gradient descent, we initialise $\{\theta_j^0\}_{j=1}^N$ as independent samples from $\mu_0$ at time $t = 0$ and then update $\{\theta_j^t\}_{j=1}^N$ deterministically, via the coupled system of \acp{ODE}
\begin{align}
\frac{\mathrm{d}\theta_i^t}{\mathrm{d}t} & = \frac{1}{N} \sum_{j=1}^N k(\theta_i^t , \theta_j^t) (\nabla \log q_0 - \vargrad  \mathcal{L}(Q_N^t))(\theta_j^t) + \nabla_1 k(\theta_j^t , \theta_i^t) , \; \; Q_N^t := \frac{1}{N} \sum_{j=1}^N \delta_{\theta_j^t}
\label{eq: gen svgd odes}
\end{align}
up to a time horizon $T$.
The first consistency result in this context was established in Proposition 3 of \citet{chazal2025computable}, who called the algorithm \emph{variational gradient descent} (\acs{VGD}).
For $\QBayes$, \ac{VGD} coincides with \ac{SVGD} \citep{liu2016stein} and the sharpest convergence analysis of \ac{SVGD} to-date appears in \citet{banerjee2025improved}.
Unfortunately the assumptions made in \citet{chazal2025computable} are too restrictive to handle $\QPC$.
To resolve this issue, a novel convergence guarantee for \ac{VGD} in the context of $\QPC$ is presented next.

\begin{remark}[Other numerical methods for $\QPC$]
\label{rem: other numerics}
One could approximate $\QPC$ using established numerical methods designed for variational tasks, the most well-studied of which is mean-field Langevin dynamics.
However, the seismic tomography community have reported empirical evidence in favour of gradient flow-type algorithms such as \ac{SVGD} over numerical methods based on long-run ergodic averages \citep[see e.g.][]{zhang2020seismic,zhang20233}.
Here we propose to approximate both $\QBayes$ and $\QPC$ using \ac{VGD}, noting that in the first case \ac{VGD} coincides with \ac{SVGD} due to the linear form of the loss function $\LBayes$.
In fact, we will see in \Cref{sec: implement} that applying \ac{VGD} to $\QPC$ requires only a one line-change to standard \ac{SVGD}.
\end{remark}

\subsection{Consistency of \ac{VGD}}
\label{sec: theory}

To discuss the consistency of \ac{VGD} we need to specify in what sense the approximation is consistent.
To this end, let $\mathcal{B}_{k}^d = \{v \in \mathcal{H}_k^d : \sum_i \|v_i\|_{\mathcal{H}_k}^2 \leq 1\}$ denote the unit ball in $\mathcal{H}_k^d$.
Let $\mathcal{L}^1(Q) := \{f : \mathbb{R}^d \rightarrow \mathbb{R} : \int \|f(x)\| \; \mathrm{d}Q(x) < \infty  \}$ denote\footnote{This should not be confused with the notation $\mathcal{L}$ for the loss function in \eqref{eq: objective}.} the set of $Q$-integrable functions on $\mathbb{R}^d$.
Let $f_{-}$ denote the negative part $f_{-}: x \mapsto \min\{0 , f(x)\}$ of a function $f$. 
The \acf{KGD} \citep[][Definition 4]{chazal2025computable} 
\begin{align}
	\mathrm{KGD}_k(Q) 
        := \sup_{\substack{ v \in \mathcal{B}_k^d \ \text{\normalfont s.t.} \\ (\mathcal{T}_{Q}v)_{-} \in \mathcal{L}^1(Q) } } \left\lvert \int \mathcal{T}_{Q} v(\theta) \; \mathrm{d} Q(\theta) \right\rvert \label{eq: KGD}
\end{align}
can be interpreted as a gradient norm for $\mathcal{J}$ using \eqref{eq: KL to KGD}; a small \ac{KGD} indicates that $Q$ is close to being a stationary point (and in particular a minimiser, due to convexity) of $\mathcal{J}$.
The precise topologies induced on $\mathcal{P}(\mathbb{R}^d)$ by \ac{KGD} can be weaker or stronger depending on how the kernel $k$ is selected; this is discussed in detail in \citet{chazal2025computable} but such discussion is beyond the scope of the present work.

Let $C^r(\mathbb{R}^d)$ denote the set of $r$ times continuously differentiable functionals on $\mathbb{R}^d$.
The proof of the following result is contained in \Cref{app: proof of VGD}:

\begin{theorem}[Consistency of \ac{VGD} for $\QPC$] \label{thm: main text}
Assume that:
\begin{enumerate}
    \item[(i)] \emph{Initialisation:} %
    $\mu_0$ has bounded support, and has a density that is bounded.
    \item[(ii)] \emph{Kernel:} $k$ is $C^2(\mathbb{R}^d \times \mathbb{R}^d)$ with the growth of $(\theta,\vartheta) \mapsto \| \nabla_1 k(\theta,\vartheta) \|$ at most linear, $\sup_{\theta} | \Delta_1 k (\theta,\theta) | < \infty$, and $\sup_\theta | \nabla_1 k(\theta,\theta) \cdot \nabla \log q_0(\theta) | < \infty$. 
    \item[(iii)] \emph{Regularisation:} $\log q_0 \in C^2(\mathbb{R}^d)$ with the growth of $(\theta,\vartheta) \mapsto |k(\theta,\vartheta)| \| \nabla \log q_0(\vartheta) \|$ at most linear, and $\sup_{\theta } k(\theta,\theta) | \Delta \log q_0(\theta) | < \infty$.
    \item[(iv)] \emph{Regularity of $P_\theta$:} \label{asm: reg ptheta} $\theta \mapsto p_\theta(y_i | x_i)$ is positive, bounded and $C^2(\mathbb{R}^d)$ with
    \begin{enumerate}
        \item[a.] $\displaystyle \sup_{\theta } \sqrt{k(\theta,\theta)} \frac{ \| \nabla_\theta p_\theta(y_i|x_i) \| }{p_\theta(y_i | x_i)} < \infty$
        \item[b.] $\displaystyle \sup_{\theta } k(\theta,\theta) \frac{ | \Delta_\theta p_\theta(y_i|x_i) | }{ p_\theta(y_i | x_i) } < \infty$
        \item[c.] $\displaystyle \sup_{\theta } \| \nabla_1 k(\theta,\theta) \| \frac{ \| \nabla_\theta p_\theta(y_i|x_i) \| }{p_\theta(y_i | x_i)} < \infty$
    \end{enumerate}
    for each $(x_i,y_i)$ in the dataset.
\end{enumerate}
Then the dynamics defined in \eqref{eq: gen svgd odes} with $\mathcal{L} = \LPC$ satisfies
    \begin{align*}
        \frac{1}{T} \int_0^T \mathbb{E}[ \mathrm{KGD}_k^2(Q_N^t) ] \; \mathrm{d}t \leq \frac{ C_0 }{T} + \frac{C_k}{N}
    \end{align*}
    for some constants $C_0$ and $C_k$ which are $N$ and $T$-independent. 
\end{theorem}

\noindent \Cref{thm: main text} provides the first consistency guarantee for \ac{VGD} in the setting of $\LPC$.
The conditions in \Cref{thm: main text} involving $\nabla_1 k(\theta,\theta)$ are automatically satisfied when $k$ is translation-invariant, such as the Gaussian and inverse multi-quadric kernels used in \Cref{sec: empirical}, since then $\theta \mapsto k(\theta,\theta)$ is constant and $\nabla_1 k(\theta,\theta) = 0$.
For such bounded kernels the growth condition in (iii) reduces to at most linear growth of $\theta \mapsto \|\nabla \log q_0(\theta)\|$.

\begin{remark}[On the assumptions]
\label{rem: assumptions}
Our assumptions in \Cref{thm: main text} rule out models for which the \emph{Fisher score} $\theta \mapsto \nabla_\theta \log p_\theta(y_i | x_i)$ is unbounded when the kernel is translation-invariant.
This is a strong assumption in general, but it is likely to be satisfied in applications that are physics-constrained, since physical considerations often dictate that model parameters can be confined to a compact set (and then reparametrised back to $\mathbb{R}^d$), as done for seismic tomography in \citet{zhang20233}.
Establishing the consistency of \ac{VGD} under weaker assumptions would be an interesting direction for further work.
\end{remark}

\begin{algorithm}[t!]
\caption{Variational Gradient Descent for $\QBayes$ and $\QPC$} \label{alg: svgd}
\begin{algorithmic}
\Require $\{\theta_j^0\}_{j=1}^N \subset \mathbb{R}^d$ (initial particles), $\epsilon > 0$ (step size)
\For{$t = 0,\dots,T-1$}
    \State $w_\theta^t(y_i | x_i) := p_{\theta}(y_i | x_i) / ( \frac{1}{N} \sum_{r=1}^N p_{\theta_r^t}(y_i | x_i) )$
    \State $s_t(\theta) := \left\{ \begin{array}{ll} (\nabla \log q_0)(\theta) + \sum_{i=1}^n \nabla_\theta \log p_\theta(y_i | x_i) & \text{to target $\QBayes$} \\
    (\nabla \log q_0)(\theta) + \sum_{i=1}^n w_\theta^t(y_i | x_i) \nabla_\theta \log p_\theta(y_i | x_i)  & \text{to target $\QPC$} \end{array} \right. $
\For{$j = 1,\dots,N$} 
    \State $\theta_j^{t+1} \gets \theta_j^t + \frac{\epsilon}{N} \sum_{r=1}^N \nabla_1 k(\theta_r^t , \theta_j^t) + s_t(\theta_r^t) k(\theta_r^t , \theta_j^t) $ 
\EndFor 
\EndFor
\end{algorithmic}
\end{algorithm}

\begin{remark}[Implementation of VGD]
For the purposes of this paper, computation of both $\QBayes$ and $\QPC$ was performed using a time discretisation of \ac{VGD} as described in \Cref{alg: svgd}.
In each case, \Cref{alg: svgd} starts by initialising $N$ particles and then iteratively updating the particles according to an Euler discretisation of the \acp{ODE} \eqref{eq: gen svgd odes} with step size $\epsilon > 0$ to be specified.
After $T$ time steps, the collection of particles represents an empirical approximation to either $\QBayes$ or $\QPC$.
It is worth reiterating that computation of $\QPC$ requires a one-line change to existing implementations of \ac{SVGD}, and that $\QBayes$ and $\QPC$ can be computed in parallel.
The rapid convergence of \ac{VGD} in a toy two-dimensional setting is displayed in \Cref{fig: convergence}.
\end{remark}

\subsection{Testing for Misspecification}
\label{sec: implement}

Our approach, in a nutshell, is to see if both $\QBayes$ and $\QPC$ appear to be converging to the same limit or not.
If these distributions are substantially different, we interpret this as evidence that the model may be misspecified.

An obvious question at this point is \emph{how to decide whether $\QBayes$ and $\QPC$ are converging to the same limit}?
If the model is well-specified, $\QBayes$ concentrates around the true parameter $\theta_\star$.
Thus we are interested in whether $\QPC$ also appears to concentrate around $\theta_\star$ or not.
Unfortunately, it is not the case that $\QBayes$ and $\QPC$ concentrate at the same rate; it is well-known that $\QBayes$ concentrates at a rate $n^{-1/2}$, while it appears that a slower rate is typical for $\QPC$.
The lack of a complete understanding of the concentration of $\QPC$ limits the extent to which the above question can be answered. 
Instead, we propose to compare the predictive distributions 
$$
P_{\QBayes}(\cdot | x) = \int P_\theta(\cdot | x) \; \mathrm{d}\QBayes(\theta) \qquad \text{and} \qquad  P_{\QPC}(\cdot | x) = \int P_\theta(\cdot | x) \; \mathrm{d}\QPC(\theta) ,
$$ 
which for simplicity we will denote in shorthand as $\PBayes$ and $\PPC$, with the dependence on $x$ left implicit. 
In the well-specified case, for a suitable discrepancy $\mathcal{D}_n$ (which we will see later can be weakly $n$-dependent),
we should expect that $\mathcal{D}_n( \PPC , \PBayes)  \rightarrow 0$ in an appropriate sense as $n \rightarrow \infty$ \citep[see][Theorem 1]{mclatchie2025predictively}.
That is, if the number of data $n$ is large enough then $\PPC$ and $\PBayes$ should be almost identical when the model is well-specified.
Conversely, if $\QPC$ does not concentrate around the same parameter $\theta_\star$ as $\QBayes$, then we can expect to detect this as an irreducible difference between the predictive distributions $\PPC$ and $\PBayes$.

\begin{figure}[t!]
    \includegraphics[width=\textwidth]{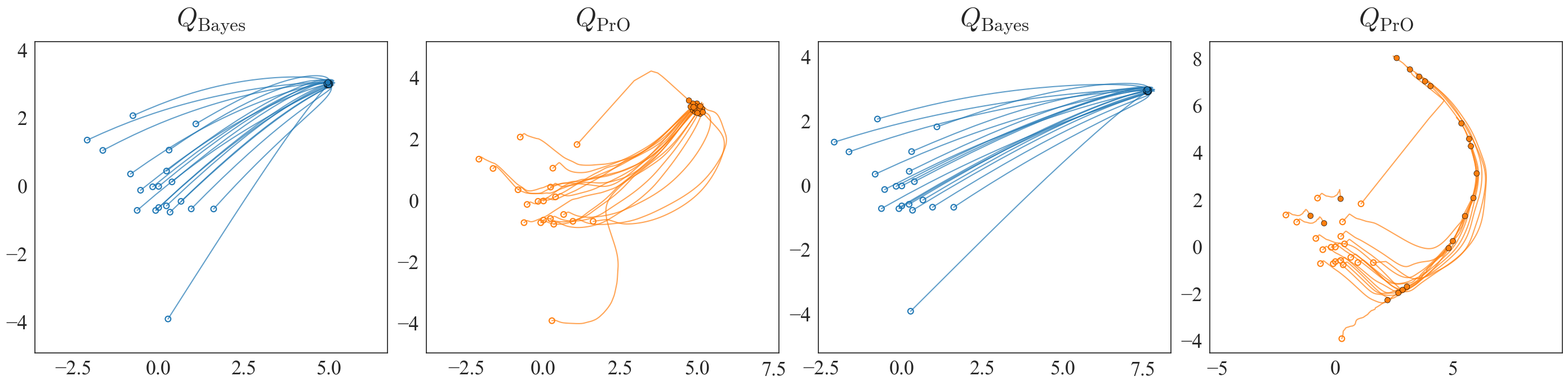}
    \caption{Illustrating the convergence of \acf{VGD} in the context of the two-dimensional example discussed in \Cref{sec: sims}.
    The hollow circular markers depict initial particle locations $\{\theta_j^0\}_{j=1}^N$, lines represent trajectories at intermediate times $\{\theta_j^t\}_{j=1}^N$, and filled circular markers depict final locations $\{\theta_j^T\}_{j=1}^N$.
    The left and centre-left panels correspond to $\QBayes$ and $\QPC$ in a setting where the statistical model is well-specified, while the centre-right and right panels correspond to $\QBayes$ and $\QPC$ in a setting where the statistical model is misspecified.
    Here $N = 20$ particles were used.
    \alttext{A figure with four panels.  Each panel displays the trajectories followed by $N = 20$ particles under the VGD dynamics, where the posterior is either the standard Bayesian posterior or the PrO posterior, and when the statistical model is well-specified and misspecified.}
    }
    \label{fig: convergence}
\end{figure}

To operationalise this idea, suppose that $\{y_i\}_{i=1}^n \subset \mathbb{R}^{p}$ for some $p \in \mathbb{N}$ and let $\kappa : \mathbb{R}^{p} \times \mathbb{R}^{p} \rightarrow \mathbb{R}$ be a symmetric positive semi-definite kernel.
Then we propose to construct an approximate null distribution for the (average squared) \ac{MMD} test statistic associated with $\kappa$ \citep{smola2007hilbert}, i.e.
\begin{align}
    \mathcal{T}(\{(x_i,y_i)\}_{i=1}^n) & \equiv \mathcal{D}_n( \PPC , \PBayes) := \frac{1}{n} \sum_{i=1}^n \mathrm{MMD}_\kappa^2( \PPC(\cdot | x_i) , \PBayes(\cdot | x_i) ),  \label{eq: MMD}
\end{align}
under the hypothesis that the statistical model is well-specified.
To do so, we let $\theta_n$ be any strongly consistent estimator of $\theta_\star$ (for the experiments we report, we take $\theta_n$ to be the mean of $\QBayes$), and use a parametric bootstrap; i.e. repeatedly compute $\mathcal{T}(\{(x_i,\tilde{y}_i)\}_{i=1}^n)$ based on synthetic datasets where $\tilde{y}_i$ is simulated from $P_{\theta_n}(\cdot | x_i)$.
The actual value of the test statistic \eqref{eq: MMD} can then be compared to the bootstrap null distribution so-obtained. 
{Theoretical support for this parametric bootstrap is provided in \Cref{subsec: bootstrap}.}

\begin{remark}[{Insensitivity to the learning rate}]
\label{rem: lr sensitivity}
{Although no learning rate is required in \eqref{eq: def QPC}, it may be of interest to assess the sensitivity of the test to the relative weighting of the two terms in \eqref{eq: objective}.
To this end, for $\lambda > 0$ one can replace \eqref{eq: objective} by $\mathcal{J}_\lambda(Q) := \mathcal{L}(Q) + \lambda \, \KLD(Q || Q_0)$, for both $\LBayes$ and $\LPC$.
In particular, with $\lambda=1$, the \ac{PC} objective differs from the Bayesian objective only through the interchange of the logarithm and expectation in the data-dependent loss, with the difference given by the non-negative Jensen gap.
A sensitivity analysis can then be performed by computing the bootstrap $p$-value $p_\lambda$ for different choices of $\lambda$.
Our findings, reported in in \Cref{app: lr}, show that the $p$-values $p_\lambda$ are remarkably insensitive to even an order-of-magnitude variation away from our recommended default $\lambda = 1$ for each of the well-specified and misspecified settings that we considered.
}
\end{remark}

\begin{remark}[Comparison to posterior predictive checks]
\label{rem: pp check}
A \emph{posterior predictive check} compares the real dataset to synthetic datasets generated from $\PBayes$ \citep{rubin1984bayesianly,gelman1996posterior}.
Unlike our approach, no assumption that data are (conditionally) independent is needed to conduct a posterior predictive check.
The approach can also be extended to a formal hypothesis test, using the parametric bootstrap to construct an empirical null in an analogous manner to that which we have described\textcolor{red}{~\citep{hjort2006post}}.
However, upon failing a posterior predictive check, one may be left with limited guidance on how to build a more suitable model; at best we might hope to compare the posterior predictive performance of models from a given candidate set \citep{moran2023posterior}.
In contrast, if the statistical model is deemed to be misspecified, we immediately have the option of adopting $\PPC$ instead of $\PBayes$ as a viable predictive model (cf. \Cref{sec: discuss}).
\end{remark}

\begin{remark}[Comparison to testing with mixtures]
An approach to selecting among a collection of models $\mathfrak{M}_i$, proposed in \citet{kamary2014testing}, is to first fit a mixture $\sum_i w_i \mathfrak{M}_i$ and then consider the largest learned mixture weights $w_i$ as the basis for selecting a best model.
Though the authors focused on the setting where the true model is contained within the candidate model set, in the setting where all models are misspecified, it could occur that a non-trivial mixture is learned.
This is similar in spirit to our approach if each $\mathfrak{M}_i$ represents an instance of the same original statistical model; however, fitting a mixture model is associated with substantial statistical and computational difficulties (cf. \Cref{rem: mixture}), which our approach is able to avoid.
\end{remark}

\subsection{Consistency of the Parametric Bootstrap Null}
\label{subsec: bootstrap}

This section presents sufficient conditions under which the parametric bootstrap, described in \Cref{sec: implement}, {yields a test whose Type~I error is asymptotically correct}.
To state these conditions we need to be explicit about how the data are generated under the statistical model.
To this end, let $\PBayes^{\theta,u}$ and $\PPC^{\theta,u}$ respectively denote the Bayesian and \ac{PC} predictive distributions based on the dataset arising from a parametrised \emph{generator}
\begin{align}
\left[ \begin{array}{c} y_1 \\ \vdots \\ y_n \end{array} \right] = \left[  \begin{array}{c} G(\theta, x_1, u) \\ \vdots \\ G(\theta,x_n,u) \end{array} \right] ,  \label{eq: data generation}
\end{align}
where $u \sim \nu$ is a random seed drawn from an appropriate reference distribution $\nu${, and $x_1,\dots,x_n$ are the covariates.
Since the parametric bootstrap re-uses the observed covariates, throughout this section the covariates are held fixed, and all statements should be understood to hold for the given covariates $x_1,\dots,x_n$.
Let
\begin{align*}
\mathcal{T}_n(\theta,u) := \mathcal{D}_n(\PPC^{\theta,u} , \PBayes^{\theta,u})
\end{align*}
denote the test statistic \eqref{eq: MMD} computed from the dataset \eqref{eq: data generation}.
If the statistical model is well-specified with true parameter $\theta_\star$, the observed value of the test statistic is $\mathcal{T}_n(\theta_\star,u)$ for some $u \sim \nu$, while the parametric bootstrap samples from the distribution of $\mathcal{T}_n(\theta_n,U)$, where $U \sim \nu$ is independent of the data and $\theta_n$ is held fixed at the value computed from the data.
The test rejects the null hypothesis at level $\alpha$ when the observed value of the test statistic exceeds the $(1-\alpha)$-quantile $\hat{c}_n$ of this bootstrap distribution.
The parametric bootstrap is therefore justified if the distribution of $\mathcal{T}_n(\theta,U)$ varies continuously with $\theta$, in a manner that is uniform in $n$, and if $\theta_n$ is close enough to $\theta_\star$ relative to the scale on which $\mathcal{T}_n$ fluctuates.

To make this precise, let $\mathcal{H}_\kappa$ denote the \ac{RKHS} of the kernel $\kappa$, and let $\mu_P := \int \kappa(y,\cdot) \, \mathrm{d}P(y) \in \mathcal{H}_\kappa$ denote the kernel mean embedding of a distribution $P$, so that $\mathrm{MMD}_\kappa(P,P') = \|\mu_P - \mu_{P'}\|_{\mathcal{H}_\kappa}$ \citep{smola2007hilbert}.
Our analysis is restricted to regression models that are \emph{linear in the parameters}, which include two of the three regression tasks considered in \Cref{sec: sims}; the general case is discussed after \Cref{cor: type I}.
The proof of the following result is contained in \Cref{app: bootstrap proof}:
}

\begin{theorem}[{Stability of the Test Statistic for Linear Models}]
\label{thm: bootstrap}
{Assume that:
\begin{enumerate}
    \item[(i)] \emph{Linear model with additive noise:} The generator takes the form $G(\theta,x,u) = A(x) \theta + e(x,u)$, for some matrices $A(x) \in \mathbb{R}^{p \times d}$, where for each $x$ the distribution of $e(x,u)$, $u \sim \nu$, has a positive, bounded and continuous density $g_x$, and correspondingly $p_\theta(y|x) = g_x(y - A(x) \theta)$.
    \item[(ii)] \emph{Prior with Lipschitz log-density:} $Q_0$ admits a positive density $q_0$ on $\mathbb{R}^d$ with $|\log q_0(\theta) - \log q_0(\vartheta)| \leq L_0 \|\theta - \vartheta\|$ for some $L_0 \geq 0$ and all $\theta$, $\vartheta$.
    \item[(iii)] \emph{Bounded translation-invariant kernel:} $\kappa(y + c, y' + c) = \kappa(y,y')$ for all $y$, $y'$, $c \in \mathbb{R}^p$, and $M^2 := \sup_y \kappa(y,y) < \infty$.
\end{enumerate}
Then, for all $\theta$, $\vartheta$, $u$, all covariates $x_1,\dots,x_n$, and all $n$,
    \begin{align*}
        \left| \sqrt{\mathcal{T}_n(\vartheta,u)} - \sqrt{\mathcal{T}_n(\theta,u)} \right| \leq K \|\vartheta - \theta\| , \qquad K := 4 M L_0 .
    \end{align*}
In particular, the $(1-\alpha)$-quantiles of $\sqrt{\mathcal{T}_n(\vartheta,U)}$ and $\sqrt{\mathcal{T}_n(\theta,U)}$, $U \sim \nu$, differ by at most $K \|\vartheta - \theta\|$.}
\end{theorem}

{
\noindent The constant $K$ in \Cref{thm: bootstrap} depends only on the prior and the kernel; it does not depend on $n$, on the covariates, on the random seed, on the matrices $A(x)$, or on the distribution of the noise.
Thus the bootstrap null distribution is within $K \|\theta_n - \theta_\star\|$ of the true null distribution at the same sample size $n$, on the scale of $\sqrt{\mathcal{T}_n}$.
This leads to the following sufficient condition for asymptotic Type~I error control, whose proof is also contained in \Cref{app: bootstrap proof}:

\begin{corollary}[Asymptotic Type~I Error Control]
\label{cor: type I}
Suppose that the statistical model is well-specified with true parameter $\theta_\star$, and that the assumptions of \Cref{thm: bootstrap} hold.
Let $r_n > 0$ be a sequence such that $r_n \mathcal{T}_n(\theta_\star,U) \stackrel{d}{\rightarrow} F$, $U \sim \nu$, where $F$ is continuous and strictly increasing at its $(1-\alpha)$-quantile.
Let $\theta_n$ be an estimator of $\theta_\star$, computed from the observed data, such that $\sqrt{r_n} \|\theta_n - \theta_\star\| \stackrel{p}{\rightarrow} 0$.
Then the parametric bootstrap test has Type~I error 
\begin{align*}
\mathbb{P}\left( \mathcal{T}_n(\theta_\star,u) > \hat{c}_n \right) \rightarrow \alpha 
\end{align*}
as $n \rightarrow \infty$, where the randomness is with respect to the random seed $u \sim \nu$ that generated the observed data.
\end{corollary}

\noindent For a root-$n$ consistent estimator $\theta_n$, such as the posterior mean of $\QBayes$ under standard regularity conditions, the rate condition in \Cref{cor: type I} requires that $r_n = o(n)$; that is, $\sqrt{\mathcal{T}_n}$ must vanish more slowly than $n^{-1/2}$ under the null.
Whether this holds depends on the rate at which $\QPC$ concentrates relative to $\QBayes$, and on how this carries over to the predictive distributions, which is not yet well understood \citep{mclatchie2025predictively}.
Establishing when the conditions of \Cref{cor: type I} hold is therefore left for future work.

The assumptions of \Cref{thm: bootstrap} are satisfied by the Gaussian regression model \eqref{eq: Gaussian location} whenever the regression function is linear in the parameters, i.e. $f_\theta(x) = A(x) \theta$, together with the Gaussian kernel $\kappa$ (for which $M = 1$) with a fixed lengthscale, and any prior whose log-density is Lipschitz, such as a multivariate Student-$t$ or Laplace-type prior.
This includes the quadratic and linear regression tasks in \Cref{sec: sims}.
Note that a Gaussian prior is not covered by assumption (ii), since its log-density is not Lipschitz.
Extending the analysis beyond models that are linear in the parameters, such as the sigmoid regression task in \Cref{sec: sims} and the seismic tomography case study in \Cref{sec: waveform}, is left for future work.
}

\section{Empirical Assessment}
\label{sec: empirical}

Preempting our application to seismic tomography, our focus in this section is on Gaussian regression models of the form
\begin{align}
p_\theta(y_i | x_i) = \frac{1}{\sqrt{ \det(2 \pi \Sigma) } } \exp\left( - \frac{1}{2} \| \Sigma^{-\frac{1}{2}}(y_i - f_\theta(x_i)) \|^2 \right), \label{eq: Gaussian location}
\end{align}
for responses $\{y_i\}_{i=1}^n$ conditional on covariates $\{x_i\}_{i=1}^n$, with known measurement error covariance matrix $\Sigma$ and unknown regression parameters $\theta \in \mathbb{R}^d$.
Our interest is in settings where the parametric regression function $f_\theta(x)$ could be misspecified.
Proceeding with Bayesian inference would be problematic if $f_\theta(x)$ is indeed misspecified, since increasing the number of data $n$ would cause the posterior to collapse onto a single `best' parameter $\theta_\star$ and the predictions from the model will collapse also to $f_{\theta_\star}$, i.e. the predictions are simultaneously high-confidence and incorrect.

Our principal interest is in whether a comparison of $\QBayes$ and $\QPC$ based on a sufficiently large dataset enables misspecification to be detected.
\Cref{sec: sims} reports a detailed empirical investigation in a controlled setting where data are generated from a simple known regression model, while a challenging application to seismic tomography is presented in \Cref{sec: waveform}.
In both cases the statistical model takes the form \eqref{eq: Gaussian location}, and we can interpret the sufficient conditions for the convergence of \ac{VGD} in \Cref{thm: main text} this context.
The proof of the following result is contained in \Cref{app: gauss location}:

\begin{proposition}[Regularity conditions for Gaussian regression model \eqref{eq: Gaussian location}] \label{prop: Gauss location MT}
    Let $P_\theta$ be the Gaussian regression model in \eqref{eq: Gaussian location}, and let $k$ be a symmetric positive semi-definite kernel for which $f_\theta(x_i)$, $\sqrt{k(\theta,\theta)} \nabla_\theta f_\theta(x_i)$, $\| \nabla_1 k(\theta,\theta) \| \| \nabla_\theta f_\theta(x_i) \|_{\mathrm{op}}$, and $k(\theta,\theta) \Delta_\theta f_\theta(x_i)$ are bounded in $\theta$ for each $x_i$ in the dataset. 
    Then condition (iv) of \Cref{thm: main text} is satisfied.
\end{proposition}

\noindent Although the assumptions of \Cref{prop: Gauss location MT} are too strong for applications such as linear regression, where the regression function $f_\theta(x) = \theta \cdot x$ can be unbounded, for many physics-constrained inverse problems the boundedness requirements are likely to be satisfied.

\subsection{Simulation Study}
\label{sec: sims}

The aim of this section is to empirically assess whether our methods can detect when the regression function $f_\theta(x)$ is misspecified.
To this end, we considered several toy regression tasks of the form \eqref{eq: Gaussian location}, varying both the size of the dataset, the dimension of the parameter vector, and the extent to which the data-generating distribution departs from the regression model.
Computation for these toy models is straightforward, so the formal test for model misspecification proposed in \Cref{sec: implement}, based on the parametric bootstrap, is practical; this is in contrast to the seismic tomography case study in \Cref{sec: waveform}.
Full details of our simulation setup are reserved for \Cref{app: sims}.

Results are presented in \Cref{fig: sim study}, where each row corresponds to a different regression model.
It is visually apparent that the spread of the posterior predictive $\PPC$ is similar to that of $\PBayes$ when the model is well-specified, and much wider than the spread of $\PBayes$ when the model is misspecified.
This difference in the misspecified setting occurs because the standard Bayesian posterior $\QBayes$ is destined to concentrate on a single `least bad' parameter $\theta_\star$ as the size of the dataset is increased, while the predictively oriented posterior $\QPC$ is able to adapt to the level of model misspecification, resulting in an irreducible uncertainty in $\QPC$ that does not vanish as the number of data is increased. 
The parameter posteriors $\QBayes$ and $\QPC$ themselves for these examples are presented in \Cref{fig: sim study 2} of \Cref{app: additional}, where we also verify the convergence of the \ac{VGD} algorithm used to obtain these results, as measured using the \ac{KGD} in \eqref{eq: KGD}.

The distribution of the test statistic $\mathcal{T}$ under the parametric bootstrap null described in \Cref{sec: implement}, alongside the actual realised value of $\mathcal{T}$, are also displayed in \Cref{fig: sim study}.
It can be seen that the realised value of $\mathcal{T}$ is far into the tail of the null distribution when data are not generated from the statistical model, meaning that the test statistic is able to detect that the statistical model is misspecified. 
Empirically, a larger sample size $n$ increases the power of the test, as expected; see \Cref{fig: diff data size} in \Cref{app: additional}. 
Conversely, the detection of model misspecification is more challenging when a large number of parameters $\theta$ are being estimated; see \Cref{fig: diff dimension} in 
\Cref{app: additional}.

\begin{figure}[t!]
    \includegraphics[width=\textwidth]{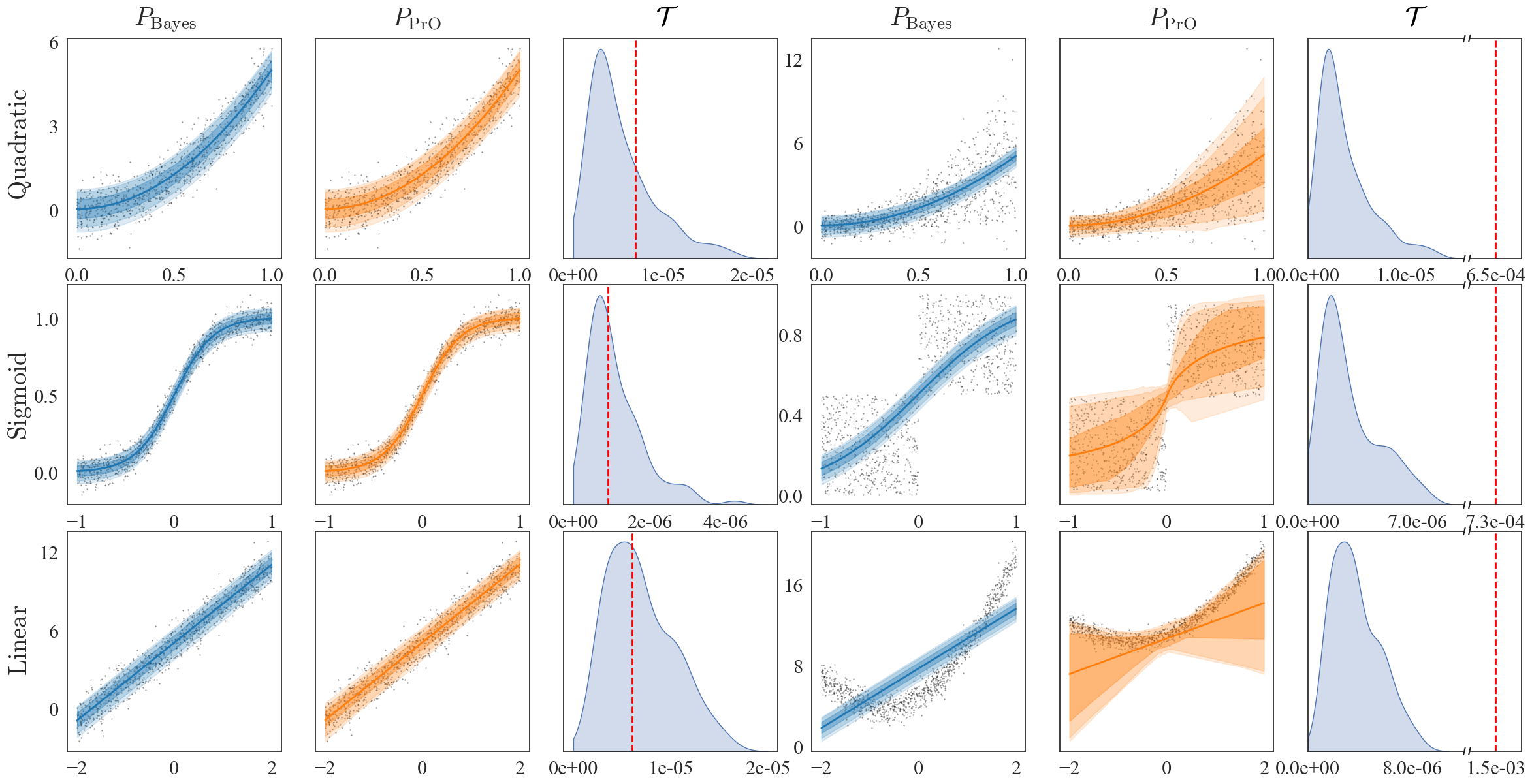}
    \caption{Simulation Study.  Each row considers a regression task in which the data are either generated from the statistical model (well-specified, left) or not generated from the statistical model (misspecified, right). The posterior predictive distributions $\PBayes(\cdot | x)$ (blue) and $\PPC(\cdot | x)$ (orange) are displayed, along with the null distribution under the hypothesis that the statistical model is well-specified, and the actual realised value of the \ac{MMD} test statistic $\mathcal{T}$ in \eqref{eq: MMD} (red dashed). 
    \alttext{A figure with 18 panels arranged into a $3 \times 6$ grid.  On each row the posterior predictive fits are displayed for both the standard and PrO posteriors, along with the distribution of the test statistic $\mathcal{T}$, when the statistical model is both well-specified and misspecified.  Each row corresponds to a different dataset and statistical model.}
    }
    \label{fig: sim study}
\end{figure}

\subsection{Bayesian Seismic Travel Time Tomography}
\label{sec: waveform}

In seismic travel time tomography, the velocity structure of a medium (e.g., the Earth's subsurface) is estimated using measured first arrival times of seismic waves propagating between source and receiver locations \citep{curtis2002probing, zhang2020seismic, zhao2022bayesian}. 
The parameter of interest is a scalar field, with $\theta(x)$ representing wave velocity\footnote{In geophysics it is traditional to refer to speed, i.e. the magnitude of the velocity, simply as \emph{velocity}.} at a spatial location $x \in \mathbb{R}^3$. 
In practice, a bounded domain $\Omega \subset \mathbb{R}^3$ is discretised into a grid and the velocity field $\theta$ is represented as a vector of values associated to each cell in the grid, i.e. the velocity is modelled as being piecewise constant.
The output of the regression model $f_\theta(x)$ represents the signal received by a seismometer at spatial location $x$, computed using a physics-constrained simulation using the velocity field $\theta$.
Typically, the model $f_\theta(x)$ is governed by the Eikonal equation $|\nabla_x f_\theta(x)| = \theta^{-1}(x)$, a high frequency approximation of the scalar wave equation, and is solved using the fast marching method \citep{rawlinson2005fast}.
The measurement error covariance matrix $\Sigma$ is assumed to be diagonal \citet{zhao2022bayesian}, with diagonal entries $\sigma_i^2$, where $\sigma_i$ is set equal to $2\%$ of the data associated to the $i^{\mathrm{th}}$ channel.

In contrast to the examples in \Cref{sec: sims}, the need to solve a partial differential equation here to evaluate the statistical model poses a computational barrier to investigating model misspecification using a hypothesis test.
As such, here we content ourselves with a qualitative exploration of whether a visual comparison of $\QBayes$ and $\QPC$ can serve as a useful diagnostic for model misspecification in this challenging context.
Of course, it is known that the regression model $f_\theta$ is to some extent misspecified relative to real world physics. 
For example, it represents a high frequency approximation of the wave physics, whereas the real case is band-limited. Additionally, travel time tomography relies only on kinematic (phase) information and ignores dynamic (amplitude) information, limiting its ability to reconstruct high resolution velocity structures of the Earth's interior. However, the key question here is whether the \emph{statistical} model is misspecified in a way that could be scientifically consequential.
To assess this, we consider a synthetic test-bed so that we have precise control of exactly how the data are misspecified.

\begin{figure}[t!]
    \centering
    \includegraphics[width=0.85\textwidth]{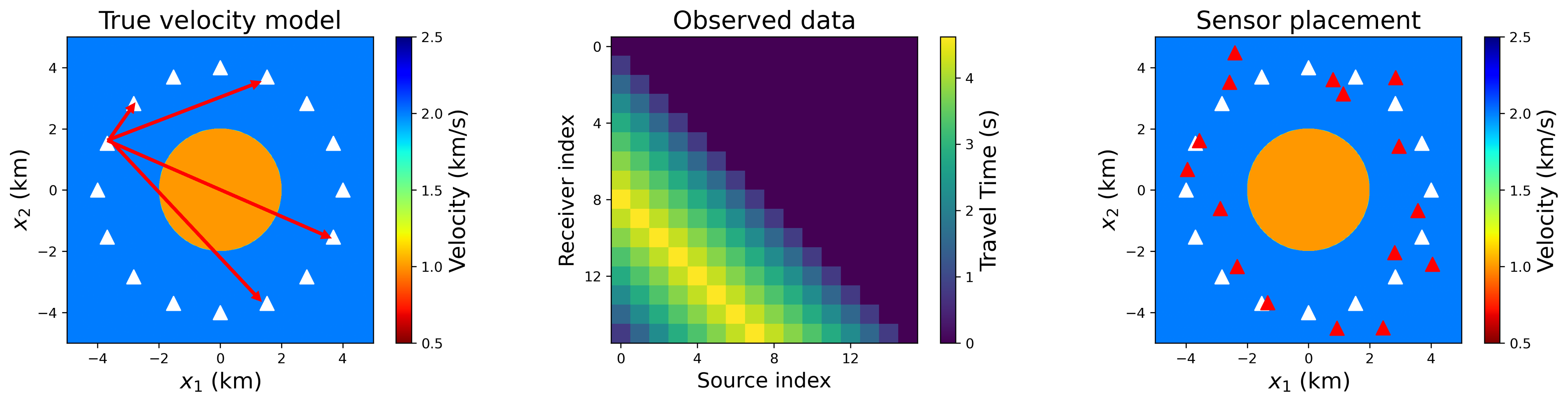}
    \caption{Seismic travel time tomography test-bed.
    Left: Data are obtained by first emitting a seismic wave from one of the $16$ sensors (white triangles) and recording the time at which the wave is detected at each of the other sensors (left panel).
    Iterating over all sensors results in $n = 256$ travel times, which are noisily observed (centre panel).
    To simulate a realistic source of model misspecification we consider a setting where the data were generated using the correct sensor placement (white triangles in the right panel) while the statistical model assumes an incorrect sensor placement (red triangles).
    \alttext{A figure consisting of 3 panels.  The first panel indicates the placement of sensors, the second panel displays the travel time dataset as a heatmap, and the third panel displays the incorrect sensor placement for the simulation regime in which the statistical model is misspecified.}
    }
    \label{fig: seismic set-up}
\end{figure}

For simplicity, our test-bed is defined for $x \in \mathbb{R}^2$ and is illustrated in the left panel of \Cref{fig: seismic set-up}.
Here the seismic velocity field $\theta$ is piecewise constant, with $\theta(x) = 1$ km/s for $\|z\| \leq 2$, and $\theta(x) = 2$ km/s for $\|z\| > 2$.
Data are obtained by first emitting a seismic wave from one of the $16$ sensors (depicted by triangles in \Cref{fig: seismic set-up}) and recording the time at which the wave is detected at each of the other sensors, as calculated using the fast marching method.
Iterating over all $16$ sensors yields $n = 256$ readings, each measured with noise governed by the covariance matrix $\Sigma$. 
Since the error model is fixed and known, the measurement error component of the statistical model is always well-specified (see the centre panel of \Cref{fig: seismic set-up}).
To simulate a realistic source of model misspecification we consider a setting where the data were generated using the correct sensor placement (white triangles in the right panel of \Cref{fig: seismic set-up}) while the statistical model assumes an incorrect sensor placement (red triangles).
This form of model misspecification is common in earthquake seismological tomography problems, in which the estimation of earthquake source locations can be inaccurate \citep{dziewonski1981determination}, affecting the subsequent seismic tomography results.

To facilitate tomographic reconstruction, the spatial domain $\Omega = [-5,5]^2$ is discretised into a $21 \times 21$ grid (the velocity field $\theta$ has dimension $d = 441$). 
The prior distribution $Q_0$ has each grid cell independent and uniformly distributed over the interval $[0.5,3]$ km/s; in practice we reparametrised $\theta$ from $[0.5,3]^{21 \times 21}$ to $\mathbb{R}^{21 \times 21}$ to avoid boundary considerations in \ac{VGD}.
Note that the data are non-informative about the region outside the convex hull of the sensors; accordingly the focus is on the  reconstruction within the interior of the convex hull.
Full details are contained in \Cref{app: seismic}.

\begin{figure}[t!]
    \centering
    \begin{subfigure}[t]{\textwidth}
    \includegraphics[width=\textwidth]{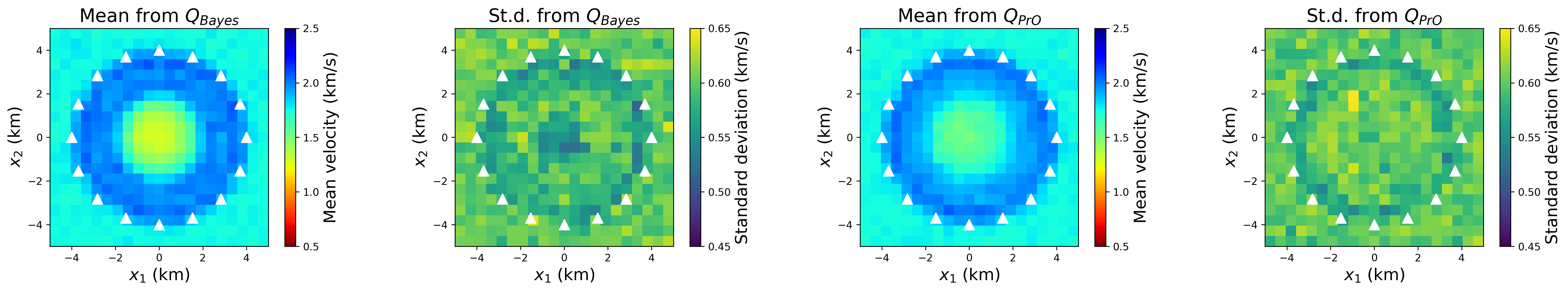}
    \caption{Well-specified sensor placement}
    \end{subfigure}
    \begin{subfigure}[t]{\textwidth}
    \includegraphics[width=\textwidth]{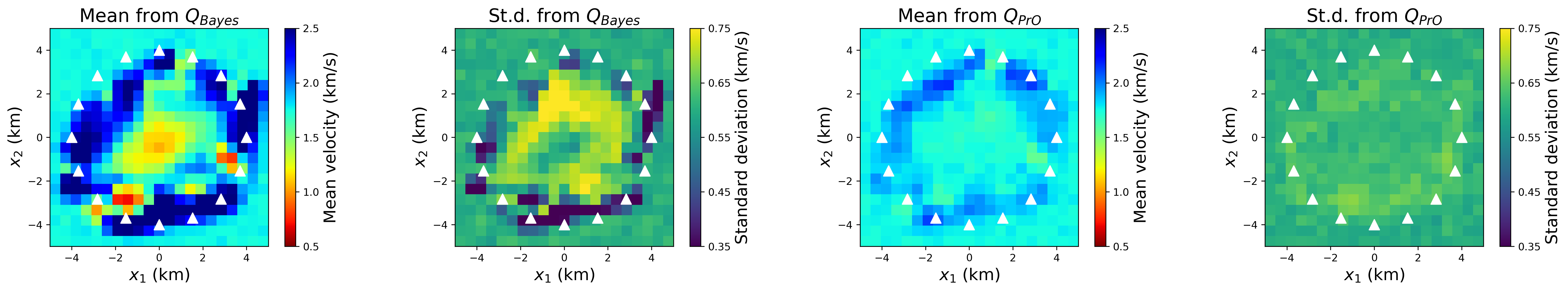}
    \caption{Misspecified sensor placement}
    \end{subfigure}
    \caption{Estimated seismic velocity $\theta$ in the setting where the sensor placement assumed in the statistical model is (a) well-specified and (b) misspecified.
    The standard Bayesian posterior $\QBayes$ (left) and the predictively oriented posterior $\QPC$ (right) are almost identical when the statistical model is well-specified, but differ substantially when the statistical model is misspecified.
    \alttext{A figure consisting of 2 rows, each row with 4 panels. On each row the parameter posterior mean and standard deviation are shown as heatmaps for both the standard Bayesian posterior and the PrO posterior.  The first row concerns a setting in which the statistical model is well-specified, while in the second row the statistical model is misspecified.  }
    }
    \label{fig: seismoc results}
\end{figure}

Results are shown in \Cref{fig: seismoc results}, where we compare pointwise means and standard deviations for $\QBayes$ and $\QPC$.
In the well-specified setting, the distributions $\QBayes$ and $\QPC$ appear almost identical; although $\QPC$ may in theory concentrate more slowly than $\QBayes$, the difference between these two `posteriors' cannot be easily distinguished, and certainly there is no evidence to suggest that these `posteriors' are not converging to the same limit.
In contrast, under model misspecification, clear differences between $\QBayes$ and $\QPC$ can be observed, both in the mean and standard deviation of the reconstructed velocity field.
Specifically, $\QBayes$ is over-confident regarding reconstruction near to the sensors, while such over-confidence is mitigated in $\QPC$.
(Curiously, $\QBayes$ also exhibits a \emph{higher} standard deviation than $\QPC$ for the central region; this indicates to us that reasoning about the impact of model misspecification on $\QBayes$ is nontrivial.)
These results are consistent with the interpretation that comparing $\QBayes$ and $\QPC$ can be an effective tool for detecting when a statistical model is misspecified.
Indeed, the computational cost of producing the diagnostic plot in \Cref{fig: seismoc results} is only a factor of $2$ larger than the cost of computing $\QBayes$ itself.

\section{Discussion}
\label{sec: discuss}

As statisticians we seek to avoid making predictions which are simultaneously highly confident and incorrect, but this scenario occurs generically in Bayesian analyses when the data are informative and the statistical model is misspecified.
To address this challenge, we combined the emerging ideas of predictively oriented inference and variational gradient descent to obtain a simple, practical and general approach to detect model misspecification in the Bayesian context.
An appealing aspect of our approach is that we do not require data splitting (e.g. as in holdout predictive checks or cross-validation) or the manual specification of alternative statistical models (as in comparative approaches).
The approach was successfully demonstrated both in simulation studies and in the challenging seismic travel time tomography context.

In settings where our methods detect model misspecification, an obvious next question is \emph{how to proceed with a misspecified model}?
This important question is outside the scope of the present work, but has been the subject of considerable research effort.
Potential solutions include nonparametric learning of the correct model \citep{kennedy2001bayesian,alvarez2013linear} and judicious use of an incorrect model \citep{bissiri2016general,knoblauch2022optimization}; a recent review is provided by \cite{nott2023bayesian}.
However, a compelling alternative that is perfectly aligned with our work is to use the predictively oriented `posterior' in place of the standard Bayesian posterior, as argued in \citet{lai2024predictive,shen2025prediction,mclatchie2025predictively}.

An important related issue is how to pinpoint which part(s) of the model are misspecified, in settings where misspecification is detected.
Indeed, this would help to facilitate iterative model improvement, in-line with the standard Bayesian workflow \citep[][Section 7]{gelman2020bayesian}.
The concept of a ``part'' only makes sense if additional structure is assumed on the statistical model; understanding how ``local'' misspecification gives rise to an irreducible uncertainty in the corresponding relevant parameters in $\QPC$ would be an interesting direction for future work.
{Finally, our theoretical analysis of the parametric bootstrap was restricted to regression models that are linear in the parameters; extending this analysis to general models represents an important direction for future work.}

\paragraph{Acknowledgments}
QL was supported by the China Scholarship Council 202408060123. 
CJO, ZS were supported by EPSRC EP/W019590/1.
CJO was supported by a Philip Leverhulme Prize PLP-2023-004.
XZ and AC thank the Edinburgh Imaging Project (\href{https://blogs.ed.ac.uk/imaging/}{EIP}) sponsors (BP and TotalEnergies) for supporting this research.
The authors are grateful to Fran\c{c}ois-Xavier Briol, Badr-Eddine Chérief-Abdellatif, David Frazier, Heishiro Kanagawa, Jeremias Knoblauch, Yann McLatchie, the Associate Editor and both Reviewers for insightful discussion during the completion of this work.

\bibliographystyle{abbrvnat}
\bibliography{bibliography}

\newpage

\appendix

\section{Proofs}

This appendix contains proofs for all theoretical results in the main text.
Preliminary results are contained in \Cref{app: prelim}, while \Cref{prop: Qmix well def} is proven in \Cref{app: well defined}, \Cref{thm: main text} is proven in \Cref{app: proof of VGD}, Proposition \ref{prop: Gauss location MT} is proven in \Cref{app: gauss location}, and \Cref{thm: bootstrap} and \Cref{cor: type I} are proven in \Cref{app: bootstrap proof}.

\paragraph{Additional Notation}

Let $b_Q(\theta) := (\nabla \log q_0)(\theta) - \vargrad  \mathcal{L}(Q)(\theta)$ for all $Q \in \mathcal{P}(\mathbb{R}^d)$ and all $\theta \in \mathbb{R}^d$.
This can be considered a $Q$-dependent generalisation of the \emph{Stein score} \citep{liu2016stein}, which is recovered in the case of linear $\mathcal{L}$.
For $f,g : \mathbb{R}^d \rightarrow \mathbb{R}$, write $f(x) \lesssim g(x)$ if there exists a finite constant $C$ such that $f(x) \leq C g(x)$ for all $x \in \mathbb{R}^d$.
For a suitably differentiable function $f$, let $\nabla^2 f$ denote the matrix of mixed partial derivatives $\partial_i\partial_j f$.

Throughout, $\|Q - Q'\|_{\mathrm{TV}}$ denotes the total variation norm of the signed measure $Q - Q'$, for which Pinsker's inequality states that $\|Q - Q'\|_{\mathrm{TV}} \leq \sqrt{2 \KLD(Q || Q')}$.

\subsection{Preliminary Results}
\label{app: prelim}

First we derive the variational gradient of $\LPC$:

\begin{proposition}[Explicit form of variational gradient] \label{prop: explicit variational grad}
    Let $\mathcal{L} = \LPC$.
    Let $\theta \mapsto p_\theta(y_i | x_i)$ be positive, bounded and differentiable for each $(x_i,y_i)$ in the dataset.
    Then 
    \begin{align}
        \vargrad \mathcal{L}(Q)(\theta) & = - \sum_{i=1}^n \frac{\nabla_\theta p_\theta(y_i | x_i)}{p_Q(y_i | x_i)} .\label{eq: var grad Lmix}
    \end{align}
\end{proposition}
\begin{proof}
    The first variation is
    \begin{align}
        \mathcal{L}'(Q)(\theta) & = - \sum_{i=1}^n \frac{p_\theta(y_i | x_i)}{p_Q(y_i | x_i)} \label{eq: first var lmix} 
    \end{align}
    which is well-defined since, under our assumptions, $p_Q(y_i | x_i) = \int p_\theta(y_i | x_i) \; \mathrm{d}Q(\theta)$ is strictly positive (so the denominator in \eqref{eq: first var lmix} is non-zero) and bounded (so \eqref{eq: first var lmix} is integrable with respect to all perturbations $\chi \in \mathcal{P}(\mathbb{R}^d)$; cf. the definition of first variation in \Cref{sec: var grads}) for all $Q \in \mathcal{P}(\mathbb{R}^d)$ and all $(x_i,y_i)$ in the dataset.
    The variational gradient is thus \eqref{eq: var grad Lmix}, where the derivatives were assumed to be well-defined.
\end{proof}

The following specialises Proposition 1 in \citealp{chazal2025computable}, which deals with general matrix-values kernels $K(\theta,\vartheta)$ to the case of a scalar kernel, i.e. $K(\theta,\vartheta) = k(\theta,\vartheta) I_{d \times d}$, and presents an explicit formula for the \ac{KGD}.

\begin{proposition}[Computable form of \ac{KGD}; special case of Proposition 1 in \citealp{chazal2025computable}]
\label{lem: computable}
    Let $Q_0$ have a density $q_0 > 0$ on $\mathbb{R}^d$.
    Let $b_Q$ be well-defined.
    Let $k$ be a symmetric positive semi-definite kernel for which $\nabla_1 k$ and $\nabla_1 \cdot \nabla_2 k$ are well-defined. 
    Suppose $\mathcal{T}_Q \mathcal{H}_k^d \subset \mathcal{L}^1(Q)$.
    Then 
    \begin{equation*}
	\mathrm{KGD}_k(Q) = \left( \iint k_Q(\theta,\vartheta) \; \mathrm{d} Q(\theta) \mathrm{d} Q(\vartheta) \right)^{1/2}  
    \end{equation*}
    where 
    \begin{align*}
        k_Q(\theta,\vartheta) := \nabla_1 \cdot \nabla_2 k(\theta,\vartheta) + \nabla_1 k(\theta,\vartheta) \cdot b_{Q}(\vartheta) + \nabla_2 k(\theta,\vartheta) \cdot b_{Q}(\theta) + k(\theta,\vartheta) b_{Q}(\theta) \cdot b_{Q}(\vartheta)
    \end{align*} 
    is a $Q$-dependent kernel. 
\end{proposition}

\noindent In particular, for $Q_N = \frac{1}{N} \sum_{i=1}^N \delta_{\theta_i}$ an empirical distribution,
\begin{align}
    \mathrm{KGD}_k^2(Q_N) = \frac{1}{N^2} \sum_{i=1}^{N} \sum_{j=1}^{N} \left\{ \begin{array}{l} \nabla_1 \cdot \nabla_2 k(\theta_i,\theta_j) + \nabla_1 k(\theta_i,\theta_j) \cdot b_{Q_N}(\theta_j) \\ \qquad + \nabla_2 k(\theta_i,\theta_j) \cdot b_{Q_N}(\theta_i) + k(\theta_i,\theta_j) b_{Q_N}(\theta_i) \cdot b_{Q_N}(\theta_j) \end{array} \right\}  \label{eq: KGD explicit}
\end{align}
which allows for \ac{KGD} to be explicitly computed.

\subsection{Proof of \Cref{prop: Qmix well def}}
\label{app: well defined}

\begin{proof}[Proof of \Cref{prop: Qmix well def}]
    Introduce the shorthand 
    $$
    \mathcal{J}(Q) = \mathcal{L}(Q) + \KLD(Q || Q_0) , \qquad \mathcal{L}(Q) = \LPC(Q) = - \sum_{i=1}^n \log p_Q(y_i | x_i) .
    $$
    Under our assumptions $\mathcal{L}$ is weakly continuous, since if $Q_n \rightarrow Q$ weakly then
    \begin{align*}
        \int p_\theta(y_i | x_i) \; \mathrm{d}Q_n(\theta) & \rightarrow \int p_\theta(y_i | x_i) \; \mathrm{d}Q(\theta)
    \end{align*}
    since the integrand is bounded and continuous, while the limit is strictly positive since the integrand is positive.
    Thus
    \begin{align*}
        \mathcal{L}(Q_n) & = - \sum_{i=1}^n \log \int p_\theta(y_i | x_i) \; \mathrm{d}Q_n(\theta) 
        \rightarrow - \sum_{i=1}^n \log \int p_\theta(y_i | x_i) \; \mathrm{d}Q(\theta) = \mathcal{L}(Q) ,
    \end{align*}
    as claimed.
    Further note that $\mathcal{L}$ is convex, since for $Q , Q' \in \mathcal{P}(\mathbb{R}^d)$ and $t \in (0,1)$,
    \begin{align*}
        \mathcal{L}( tQ + (1-t) Q') & = - \sum_{i=1}^n \log \left[ t \int p_\theta(y_i | x_i) \; \mathrm{d}Q(\theta) + (1-t) \int p_\theta(y_i | x_i) \; \mathrm{d}Q'(\theta)  \right] \\
        & \leq - t \sum_{i=1}^n \log \int p_\theta(y_i | x_i) \; \mathrm{d}Q(\theta) - (1-t) \sum_{i=1}^n \log \int p_\theta(y_i | x_i) \; \mathrm{d}Q'(\theta) \\
        & = t \mathcal{L}(Q) + (1-t) \mathcal{L}(Q')
    \end{align*}
    by convexity of $z \mapsto - \log(z)$.
    Since $\mathrm{KLD}(\cdot || Q_0)$ is weakly lower semi-continuous and strictly convex, it follows that $\mathcal{J}$ is as well.
    The rest of the proof follows a standard argument \citep[e.g. Proposition 5 in][]{hu2021mean}.
    To this end, note that
    \begin{align*}
        \mathcal{S} := \left\{ Q \in \mathcal{P}(\mathbb{R}^d) : \KLD(Q || Q_0) \leq \mathcal{J}(Q_0) - \inf_{Q' \in \mathcal{P}(\mathbb{R}^d)} \mathcal{L}(Q') \right\} .
    \end{align*}
    Here the threshold is finite, since $\mathcal{J}(Q_0) = \mathcal{L}(Q_0) < \infty$ (as $p_{Q_0}(y_i|x_i) > 0$) and $\inf_{Q'} \mathcal{L}(Q') \geq - \sum_{i=1}^n \log \sup_\theta p_\theta(y_i | x_i) > -\infty$ (as $p_Q(y_i|x_i) \leq \sup_\theta p_\theta(y_i|x_i)$ for all $Q$).
    The set $\mathcal{S}$ is a (non-empty, since $Q_0 \in \mathcal{S}$) sub-level set of the \ac{KLD} and is therefore weakly compact \citep[][Lemma 1.4.3]{dupuis2011weak}.
    Since $\mathcal{J}$ is weakly lower semi-continuous, the minimum of $\mathcal{J}$ on $\mathcal{S}$ is attained.
    Since $\mathcal{J}(Q) \geq \mathcal{J}(Q_0)$ for all $Q \notin \mathcal{S}$, the minimum of $\mathcal{J}$ on $\mathcal{S}$ coincides with the global minimum of $\mathcal{J}$.
    Since $\mathcal{J}$ is strictly convex, the minimum is unique.
    The final claim is the content of \Cref{prop: alpha moment}.
\end{proof}

\begin{proposition}[Moments for $\QPC$] \label{prop: alpha moment}
    Assume that $Q_0 \in \mathcal{P}_\alpha(\mathbb{R}^d)$ admits a positive density $q_0$ on $\mathbb{R}^d$.
    Let $\theta \mapsto p_\theta(y_i | x_i)$ be positive, bounded and continuous for each $(x_i,y_i)$ in the dataset.
    Then $\QPC \in \mathcal{P}_\alpha(\mathbb{R}^d)$.
\end{proposition}
\begin{proof}
    Following the same argument used to prove \Cref{prop: explicit variational grad}, the first variation $\LPC'$ is well-defined.
    Since $\QPC$ is the minimiser of $\mathcal{J}$, following \citet[][Corollary 2]{chazal2025computable} it is also a solution of the stationary point equation 
    \begin{align*}
        \text{cst} = \LPC'(Q) + 1 + \log \frac{\mathrm{d}Q}{\mathrm{d}Q_0}
    \end{align*}
    which implies $\QPC$ has a density $q_{\mathrm{PrO}}(\theta) \propto \exp(-\LPC'(\QPC)(\theta)) q_0(\theta)$ on $\mathbb{R}^d$.
    From \eqref{eq: first var lmix}  and our assumption, $\LPC'(\QPC)$ is bounded.
    The conclusion therefore follows from the assumption $Q_0 \in \mathcal{P}_\alpha(\mathbb{R}^d)$.
\end{proof}

\subsection{Proof of \Cref{thm: main text}}
\label{app: proof of VGD}

The main idea behind the proof of \Cref{thm: main text} is to undertake a refinement of the analysis of \ac{VGD} in \citet{chazal2025computable}; we do this in \Cref{app: refined}.
Then we verify that our refined regularity conditions hold in \Cref{app: ver reg cond}, enabling the proof of \Cref{thm: main text} to be presented in \Cref{app: proof main}.

\subsubsection{Refined Analysis of \ac{VGD}}
\label{app: refined}

The following result relaxes the conditions of Proposition 3 in \citet{chazal2025computable}, to obtain a result that is applicable in our context; further details can be found in \Cref{rem: relaxing}.
Note that Proposition 3 in \citet{chazal2025computable} is in turn a generalisation (to \ac{VGD}) of the analysis of \ac{SVGD} presented in Theorem 1 of \citet{banerjee2025improved}.

For a sufficiently regular $\mathcal{F} : \mathcal{P}(\mathbb{R}^d) \rightarrow \mathbb{R}$, let $\vargrad  \cdot \vargrad  \mathcal{F}(Q)$ denote the function $(x,y) \mapsto \sum_i \partial_i \mathcal{G}_{x,i}'(Q)(y)$ where $\mathcal{G}_x(Q) := \vargrad \mathcal{F}(Q)(x)$.

\begin{proposition}[Refined analysis of \ac{VGD}]
\label{thm: vgd converge}
Assume that:
\begin{enumerate}
    \item[(i)] \emph{Lower bound and chain rule:} $\mathcal{L}_{\min} := \inf_{Q \in \mathcal{P}(\mathbb{R}^d)} \mathcal{L}(Q) > -\infty$, $\mathbb{E}_{\theta_1,\dots,\theta_N \sim \mu_0}[ \mathcal{L}(Q_N) ] < \infty$, and the map $(\theta_1,\dots,\theta_N) \mapsto \mathcal{L}(Q_N)$ is $C^1(\mathbb{R}^{d \times N})$ with $\nabla_{\theta_i} \mathcal{L}(Q_N) = \frac{1}{N} \vargrad \mathcal{L}(Q_N)(\theta_i)$ for each $i \in \{1,\dots,N\}$, \label{assum: rho well def}
    \item[(ii)] \emph{Loss:} \label{asm: asm on loss}
    for each $Q \in \mathcal{P}(\mathbb{R}^d)$, the first variation $\mathcal{L}'(Q)$ exists and is $C^2(\mathbb{R}^d)$, and $\vargrad \cdot \vargrad \mathcal{L}(Q)$ exists;
    the map $(\bm{\theta}, \vartheta) \mapsto \vargrad \mathcal{L}(Q_N)(\vartheta)$ is $C^1(\mathbb{R}^{d \times N} \times \mathbb{R}^d)$, where $\bm{\theta} = (\theta_1,\dots,\theta_N)$ and $Q_N \equiv Q_N(\bm{\theta})$, with
    $\nabla_{\theta_i} \cdot [ \vargrad \mathcal{L}(Q_N)(\vartheta) ] = \frac{1}{N} \vargrad \cdot \vargrad \mathcal{L}(Q_N)(\vartheta)(\theta_i)$ for each $i \in \{1,\dots,N\}$ and $\vartheta \in \mathbb{R}^d$; and
    \begin{enumerate}
        \item[a.] $\sup_{Q \in \mathcal{P}(\mathbb{R}^d)} | \int k(\theta,\theta) \nabla \cdot \vargrad  \mathcal{L}(Q)(\theta) \; \mathrm{d}Q(\theta) | < \infty$, 
        \item[b.] $\sup_{Q \in \mathcal{P}(\mathbb{R}^d)} | \int \nabla_1 k(\theta,\theta) \cdot \vargrad \mathcal{L}(Q)(\theta) \; \mathrm{d}Q(\theta) | < \infty$,
        \item[c.] $\sup_{Q \in \mathcal{P}(\mathbb{R}^d)} | \iint k(\theta,\vartheta) \vargrad  \cdot \vargrad  \mathcal{L}(Q)(\theta)(\vartheta) \; \mathrm{d}Q(\theta) \mathrm{d}Q(\vartheta) | < \infty$ \label{asm: wass hess}
    \end{enumerate}
    \item[(iii)] \emph{Regularisation:} %
    $\log q_0 \in C^2(\mathbb{R}^d)$ with $\sup_{\theta} k(\theta,\theta) | \Delta \log q_0(\theta) | < \infty$,
    \item[(iv)] \emph{Initialisation:} %
    $\mu_0$ has bounded support, and has a density that is bounded.
    \item[(v)] \emph{Kernel:} $k$ is $C^2(\mathbb{R}^d \times \mathbb{R}^d)$ and satisfies both $\sup_{\theta } | \Delta_1 k (\theta,\theta) | < \infty$ and $\sup_\theta | \nabla_1 k(\theta,\theta) \cdot \nabla \log q_0(\theta) | < \infty$.   \label{asm: kernel}
    \item[(vi)] \emph{Growth:} the maps 
    $(\theta_1, \dots, \theta_N) \mapsto |k(\theta_i,\theta_j)| \| \nabla \log q_0(\theta_j) \|$, $|k(\theta_i,\theta_j)| \| \vargrad \mathcal{L}(Q_N)(\theta_j) \|$ and $\| \nabla_1 k(\theta_j,\theta_i) \|$ have at most linear growth for each $\{i,j\} \in \{1 , \dots , N\}$. \label{asm: growth}
\end{enumerate}
Then the dynamics defined in \eqref{eq: gen svgd odes} satisfies
    \begin{align*}
        \frac{1}{T} \int_0^T \mathbb{E}[ \mathrm{KGD}_k^2(Q_N^t) ] \; \mathrm{d}t \leq \frac{ C_N }{T} + \frac{C_k}{N}, \quad C_N := \KLD(\mu_0 || Q_0) + \mathbb{E}_{\theta_1,\dots,\theta_N \sim \mu_0} \left[ \mathcal{L}(Q_N) \right] - \mathcal{L}_{\min} ,
    \end{align*}
    for some constant $C_k$ which is $N$ and $T$-independent. 
\end{proposition}
\begin{proof}
The proof is organised into four steps:

\smallskip

\noindent \textit{Step 1: Flow of the particle system.}
Introduce the shorthand $\bm{\theta} := (\theta_1 , \dots , \theta_N) \in \mathbb{R}^{d \times N}$ and 
\begin{align*}
\Phi_{\bm{\theta}}(\theta_i,\theta_j) := k(\theta_i , \theta_j) \underbrace{ (\nabla \log q_0 - \vargrad  \mathcal{L}(Q_N))(\theta_j) }_{ =: b_{Q_N}(\theta_j) } + \nabla_1 k(\theta_j , \theta_i) , \qquad Q_N := \frac{1}{N} \sum_{j=1}^N \delta_{\theta_j}  ,
\end{align*}
where for convenience we have suppressed the $t$-dependence, i.e. $\theta_i \equiv \theta_i(t)$ and $Q_N \equiv Q_N(\bm{\theta})$.
Let $\bm{V} = (V_1,\dots,V_N) : \mathbb{R}^{d \times N} \rightarrow \mathbb{R}^{d \times N}$ with $V_i(\bm{\theta}) := \frac{1}{N} \sum_{j=1}^N \Phi_{\bm{\theta}}(\theta_i,\theta_j)$, so that \eqref{eq: gen svgd odes} reads $\dot{\bm{\theta}} = \bm{V}(\bm{\theta})$.
Under our assumptions $\bm{V}$ is $C^1(\mathbb{R}^{d \times N})$: indeed, $k$ and $\nabla_1 k$ are $C^1$ by (v), $\nabla \log q_0$ is $C^1$ by (iii), and $\bm{\theta} \mapsto \vargrad \mathcal{L}(Q_N)(\theta_j)$ is $C^1(\mathbb{R}^{d \times N})$ by (ii), composing $(\bm{\theta},\vartheta) \mapsto \vargrad \mathcal{L}(Q_N)(\vartheta)$ with $\bm{\theta} \mapsto (\bm{\theta},\theta_j)$.
Further, from (vi), $\Phi_{\bm{\theta}}(\theta_i,\theta_j)$ has at most linear growth as a function of $\bm{\theta}$; i.e. $\|\Phi_{\bm{\theta}}(\theta_i , \theta_j)\| \lesssim 1 + \|\bm{\theta}\|$, and hence $\|\bm{V}(\bm{\theta})\| \lesssim 1 + \|\bm{\theta}\|$.
Since $\bm{V}$ is $C^1$ with at most linear growth, solutions of $\dot{\bm{\theta}} = \bm{V}(\bm{\theta})$ exist, are unique, and are defined for all $t \in \mathbb{R}$, and the associated flow $\bm{F}_t : \mathbb{R}^{d \times N} \rightarrow \mathbb{R}^{d \times N}$, $\bm{F}_0 = \mathrm{Id}$, is a $C^1$ diffeomorphism for each $t$ (with inverse $\bm{F}_{-t}$).
Further, $(t,\bm{\eta}) \mapsto \bm{F}_t(\bm{\eta})$ and $(t,\bm{\eta}) \mapsto D\bm{F}_t(\bm{\eta})$ are continuous, and the Jacobian matrix satisfies the variational equation $\partial_t D\bm{F}_t(\bm{\eta}) = D\bm{V}(\bm{F}_t(\bm{\eta})) \, D\bm{F}_t(\bm{\eta})$ \citep[][Chapter 5, Theorem 3.1]{hartman2002ordinary}.
Let $\Lambda_t(\bm{\eta}) := \det D\bm{F}_t(\bm{\eta})$; from the variational equation and Jacobi's formula,
\begin{align}
    \partial_t \log \Lambda_t(\bm{\eta}) = \mathrm{tr}\big( D\bm{V}(\bm{F}_t(\bm{\eta})) \big) = (\nabla \cdot \bm{V})(\bm{F}_t(\bm{\eta})) , \qquad \Lambda_0(\bm{\eta}) = 1 , \label{eq: Jacobi}
\end{align}
so that $\Lambda_t(\bm{\eta}) = \exp ( \int_0^t (\nabla \cdot \bm{V})(\bm{F}_s(\bm{\eta})) \, \mathrm{d}s ) > 0$.
Since $\bm{\theta}(t) = \bm{F}_t(\bm{\theta}(0))$ and $\bm{\theta}(0)$ has density $p_N(0,\cdot) = \mu_0^{\otimes N}$, the change of variables formula shows that $\bm{\theta}(t)$ has a density $p_N(t,\cdot)$, characterised by
\begin{align}
    p_N(t,\bm{F}_t(\bm{\eta})) \, \Lambda_t(\bm{\eta}) = p_N(0,\bm{\eta}) , \qquad \bm{\eta} \in \mathbb{R}^{d \times N} . \label{eq: change of variables}
\end{align}
Let $S_0 := \operatorname{supp} p_N(0,\cdot) = (\operatorname{supp} \mu_0)^N$, which is compact by (iv).
Then $p_N(t,\cdot)$ is supported in $\bm{F}_t(S_0)$ and, for $t \in [0,T]$, these sets are contained in the compact set $\{\bm{F}_s(\bm{\eta}) : s \in [0,T], \, \bm{\eta} \in S_0\}$.
Note that no differentiability of $p_N$ will be required.

\smallskip

\noindent \textit{Step 2:  Descent on the $N$-particle free energy.}

Consider the function $\pi_N : \mathbb{R}^{d \times N} \rightarrow (0,\infty)$ defined as
\begin{align}
    \pi_N(\bm{\theta}) := q_0(\theta_1) \cdots q_0(\theta_N) \exp\left( - N \mathcal{L}(Q_N) \right) , \qquad Q_N \equiv Q_N(\bm{\theta}) , \label{eq: N particle Gibbs}
\end{align}
which is (up to normalisation) the $N$-particle Gibbs measure associated with the objective $\mathcal{J}(Q) = \mathcal{L}(Q) + \KLD(Q || Q_0)$ \citep[cf.][]{chen2022uniform}.
Crucially, $\pi_N$ is a fixed function on $\mathbb{R}^{d \times N}$; it does not depend on $t$.
From (i), $\bm{\theta} \mapsto \log \pi_N(\bm{\theta})$ is $C^1(\mathbb{R}^{d \times N})$ and its gradient can be computed exactly, with all dependence of $Q_N$ on $\bm{\theta}$ accounted for, as
\begin{align}
    \nabla_{\theta_i} \log \pi_N(\bm{\theta}) = \nabla \log q_0(\theta_i) - N \nabla_{\theta_i} \mathcal{L}(Q_N) = \nabla \log q_0(\theta_i) - \vargrad \mathcal{L}(Q_N)(\theta_i) = b_{Q_N}(\theta_i) . \label{eq: grad log piN}
\end{align}
Let
\begin{align}
H(t) := \int \log \left( \frac{p_N(t,\bm{\theta})}{ \pi_N(\bm{\theta}) } \right) p_N(t,\bm{\theta}) \; \mathrm{d}\bm{\theta} = \KLD(p_N(t,\cdot) || Q_0^{\otimes N}) + N \, \mathbb{E}\left[ \mathcal{L}(Q_N) \right] , \label{eq: def H}
\end{align}
where the expectation is with respect to $\bm{\theta} \sim p_N(t,\cdot)$.
Since $\pi_N$ does not depend on $t$, the $t$ dependence of $H(t)$ enters only through the term $p_N(t,\bm{\theta})$ in the integrand.

Write $p_0 \equiv p_N(0,\cdot)$.
Here $H(0)$ is finite; see \eqref{eq: H0 expe} below.
The change of variables $\bm{\theta} = \bm{F}_t(\bm{\eta})$, together with \eqref{eq: change of variables}, gives
\begin{align}
    H(t) & = \int \log \left( \frac{ p_0(\bm{\eta}) }{ \pi_N(\bm{F}_t(\bm{\eta})) \, \Lambda_t(\bm{\eta}) } \right) p_0(\bm{\eta}) \; \mathrm{d}\bm{\eta} \nonumber \\
    & = H(0) - \int \left[ \log \pi_N(\bm{F}_t(\bm{\eta})) - \log \pi_N(\bm{\eta}) + \log \Lambda_t(\bm{\eta}) \right] p_0(\bm{\eta}) \; \mathrm{d}\bm{\eta} , \label{eq: H change of variables}
\end{align}
with the convention $0 \log 0 = 0$.
From the chain rule, $\partial_t \log \pi_N(\bm{F}_t(\bm{\eta})) = \nabla \log \pi_N(\bm{F}_t(\bm{\eta})) \cdot \bm{V}(\bm{F}_t(\bm{\eta}))$, which together with \eqref{eq: Jacobi} shows that the $t$-derivative of the integrand in \eqref{eq: H change of variables} is $p_0(\bm{\eta})$ multiplied by a function that is continuous in $(t,\bm{\eta})$, and hence bounded on the compact set $[0,T] \times S_0$.
Differentiation under the integral is therefore justified by the dominated convergence theorem and, transforming back using \eqref{eq: change of variables},
\begin{align*}
    H'(t) & = - \int \left[ (\nabla \cdot \bm{V})(\bm{F}_t(\bm{\eta})) + \nabla \log \pi_N(\bm{F}_t(\bm{\eta})) \cdot \bm{V}(\bm{F}_t(\bm{\eta})) \right] p_0(\bm{\eta}) \; \mathrm{d}\bm{\eta} \\
    & = - \int \left[ (\nabla \cdot \bm{V})(\bm{\theta}) + \nabla \log \pi_N(\bm{\theta}) \cdot \bm{V}(\bm{\theta}) \right] p_N(t,\bm{\theta}) \; \mathrm{d}\bm{\theta} .
\end{align*}
Finally, $(\nabla \cdot \bm{V})(\bm{\theta}) = \sum_{i=1}^N \nabla_{\theta_i} \cdot V_i(\bm{\theta}) = \frac{1}{N} \sum_{i=1}^N \sum_{j=1}^N \nabla_{\theta_i} \cdot \Phi_{\bm{\theta}}(\theta_i,\theta_j)$, where $\nabla_{\theta_i}$ acts on all dependence on $\theta_i$, and from \eqref{eq: grad log piN} we have $\nabla \log \pi_N(\bm{\theta}) \cdot \bm{V}(\bm{\theta}) = \frac{1}{N} \sum_{i=1}^N \sum_{j=1}^N b_{Q_N}(\theta_i) \cdot \Phi_{\bm{\theta}}(\theta_i,\theta_j)$.
Thus
\begin{align}
    H'(t) & = - \frac{1}{N} \sum_{i=1}^N \sum_{j=1}^N \int (  \nabla_{\theta_i} \cdot \Phi_{\bm{\theta}}(\theta_i,\theta_j) + b_{Q_N}(\theta_i) \cdot \Phi_{\bm{\theta}}(\theta_i,\theta_j) ) p_N(t,\bm{\theta}) \; \mathrm{d}\bm{\theta} \nonumber \\
    & = - \mathbb{E} \left[ \frac{1}{N} \sum_{i=1}^N \sum_{j=1}^N \nabla_{\theta_i} \cdot \Phi_{\bm{\theta}}(\theta_i,\theta_j) + b_{Q_N}(\theta_i) \cdot \Phi_{\bm{\theta}}(\theta_i,\theta_j)  \right] \label{eq: derivi of KL}
\end{align}
where we have used the expectation shorthand to refer to the random initialisation of the particles.

\smallskip
\noindent \textit{Step 3:  Calculating derivatives.}
Now we aim to calculate the terms in \eqref{eq: derivi of KL}:
\begin{align*}
    \nabla_{\theta_i} \cdot \Phi_{\bm{\theta}}(\theta_i,\theta_j) & = \nabla_{\theta_i} \cdot [ k(\theta_i , \theta_j) b_{Q_N}(\theta_j)] + \nabla_{\theta_i} \cdot [ \nabla_1 k(\theta_j , \theta_i) ] \\
    & = \nabla_1 k(\theta_i,\theta_j) \cdot b_{Q_N}(\theta_j) + k(\theta_i,\theta_j) \nabla_{\theta_i} \cdot b_{Q_N}(\theta_j) + \nabla_1 \cdot \nabla_2 k(\theta_i,\theta_j) \\
    & \qquad +  \{ \nabla_1 k(\theta_i,\theta_i) \cdot b_{Q_N}(\theta_i) + \Delta_1 k(\theta_i,\theta_i)  \} \mathbbm{1}_{i=j}  \\
    b_{Q_N}(\theta_i) \cdot \Phi_{\bm{\theta}}(\theta_i,\theta_j) & = k(\theta_i,\theta_j) b_{Q_N}(\theta_i) \cdot b_{Q_N}(\theta_j) + b_{Q_N}(\theta_i) \cdot \nabla_1 k(\theta_j,\theta_i)
\end{align*}
and, writing $\Psi(\bm{\theta},\vartheta) := \vargrad \mathcal{L}(Q_N)(\vartheta)$, the multivariate chain rule gives $\nabla_{\theta_i} \cdot [\Psi(\bm{\theta},\theta_j)] = \{ (\nabla_\vartheta \cdot \Psi)(\bm{\theta},\theta_i) \} \mathbbm{1}_{i=j} + [\nabla_{\theta_i} \cdot \Psi(\bm{\theta},\vartheta)]_{\vartheta = \theta_j}$, where the first term is $\nabla \cdot \vargrad \mathcal{L}(Q_N)(\theta_i)$ and, from the second chain rule in (ii), the second term is $\frac{1}{N} \vargrad \cdot \vargrad \mathcal{L}(Q_N)(\theta_j)(\theta_i)$, so that
\begin{align*}
    \nabla_{\theta_i} \cdot b_{Q_N}(\theta_j) & =  \{ \nabla \cdot ( \nabla \log q_0 - \vargrad  \mathcal{L}(Q_N) )(\theta_i)  \} \mathbbm{1}_{i=j}  -  \frac{1}{N} \vargrad  \cdot \vargrad  \mathcal{L}(Q_N)(\theta_j)(\theta_i).
\end{align*}
Thus, collecting together terms that correspond to \ac{KGD} using \eqref{eq: KGD explicit},
\begin{align}
    H'(t) & = - N \mathbb{E} \left[ \mathrm{KGD}_k^2(Q_N) \right] + \mathbb{E} \Bigg[ \frac{1}{N^2} \sum_{i=1}^N \sum_{j=1}^N  k(\theta_i,\theta_j) \vargrad  \cdot \vargrad  \mathcal{L}(Q_N)(\theta_j)(\theta_i) \Bigg] \label{eq: remainder}   \\
    &  - \mathbb{E} \Bigg[ \frac{1}{N} \sum_{i=1}^N   k(\theta_i,\theta_i) [ \Delta \log q_0(\theta_i) - \nabla \cdot \vargrad  \mathcal{L}(Q_N)(\theta_i) ] 
    + \Delta_1 k(\theta_i,\theta_i)  + \nabla_1 k(\theta_i,\theta_i) \cdot b_{Q_N}(\theta_i) \Bigg] . \nonumber
\end{align}

\smallskip
\noindent \textit{Step 4:  Obtaining a bound.}
The final task is to bound the non-\ac{KGD} terms appearing in \eqref{eq: remainder} by a $Q_N$-independent constant.
Under our assumptions,
\begin{align*}
    & \hspace{-30pt} \left| \frac{1}{N^2} \sum_{i=1}^N \sum_{j=1}^N  k(\theta_i,\theta_j) \vargrad  \cdot \vargrad  \mathcal{L}(Q_N)(\theta_j)(\theta_i) \right| \\
    & \leq \sup_{Q \in \mathcal{P}(\mathbb{R}^d)} \left| \iint k(\theta,\vartheta) \vargrad  \cdot \vargrad  \mathcal{L}(Q)(\theta)(\vartheta) \; \mathrm{d}Q(\theta) \mathrm{d}Q(\vartheta) \right| < \infty \\
    & \hspace{-30pt} \left| \frac{1}{N} \sum_{i=1}^N k(\theta_i,\theta_i) [ \Delta \log q_0(\theta_i) - \nabla \cdot \vargrad  \mathcal{L}(Q_N)(\theta_i) ] 
    + \Delta_1 k(\theta_i,\theta_i)  + \nabla_1 k(\theta_i,\theta_i) \cdot b_{Q_N}(\theta_i) \right| \\
    & \leq \sup_{\theta} k(\theta,\theta) | \Delta \log q_0(\theta) | \\
    & \qquad + \sup_{Q \in \mathcal{P}(\mathbb{R}^d)} \left| \int k(\theta,\theta) \nabla \cdot \vargrad \mathcal{L}(Q)(\theta) \; \mathrm{d}Q(\theta) \right| + \sup_{\theta } | \Delta_1 k(\theta,\theta) | \\
    & \qquad  + \sup_\theta | \nabla_1 k(\theta,\theta) \cdot \nabla \log q_0(\theta) |  \\
    & \qquad  + \sup_{Q \in \mathcal{P}(\mathbb{R}^d)} \left| \int \nabla_1 k(\theta,\theta) \cdot \vargrad \mathcal{L}(Q)(\theta) \; \mathrm{d}Q(\theta) \right|  < \infty 
\end{align*}
which establish the required bounds, i.e.
\begin{align*}
    H'(t) & \leq - N \mathbb{E}[ \mathrm{KGD}_k^2(Q_N)  ]  + C_k 
\end{align*}
for some finite, $k$-dependent constant $C_k$.
Integrating both sides from $0$ to $T$ and rearranging yields
\begin{align*}
    \frac{1}{T} \int_0^T \mathbb{E}[ \mathrm{KGD}_k^2(Q_N) ] \; \mathrm{d}t \leq \frac{H(0) - H(T)}{N T} + \frac{C_k}{N} .
\end{align*}

It remains to bound $H(0) - H(T)$.
Since $Q_0^{\otimes N}$ is a probability distribution, the first term in \eqref{eq: def H} is non-negative, while from (i) the second term is bounded below by $N \mathcal{L}_{\min}$, so that
\begin{align*}
    H(T) \geq N \mathcal{L}_{\min} .
\end{align*}
At $t = 0$ we have $p_N(0,\cdot) = \mu_0^{\otimes N}$ and, using the tensorisation of the \ac{KLD} over product measures,
\begin{align}
    H(0) & = \KLD(\mu_0^{\otimes N} || Q_0^{\otimes N}) + N \, \mathbb{E}_{\theta_1,\dots,\theta_N \sim \mu_0}[\mathcal{L}(Q_N)] \nonumber \\
    & = N \left\{ \KLD(\mu_0 || Q_0) + \mathbb{E}_{\theta_1,\dots,\theta_N \sim \mu_0}[\mathcal{L}(Q_N)] \right\} . \label{eq: H0 expe}
\end{align}
Here $\KLD(\mu_0 || Q_0) < \infty$, since from (iv) the density $\mu_0$ is bounded, so that $\int \mu_0 \log \mu_0 \leq \log \sup_\theta \mu_0(\theta)$, while from (iii) $q_0$ is continuous and positive, and hence bounded away from zero on the compact support of $\mu_0$.
Combining these bounds, $H(0) - H(T) \leq N C_N$, as required.

\end{proof}

\begin{remark}[Relaxing the conditions of Proposition 3 in \citet{chazal2025computable}] \label{rem: relaxing}
    Compared to Proposition 3 in \citet{chazal2025computable} we have relaxed both (ii) on the loss and condition (v) on the kernel.
    In (ii) we have relaxed boundedness conditions on $\nabla \vargrad \mathcal{L}$ and $\vargrad \cdot \vargrad \mathcal{L}$ (which do not hold for $\LPC$ in our context) into integrability conditions on $\nabla \cdot \vargrad \mathcal{L}$ and $\vargrad \cdot \vargrad \mathcal{L}$ (which, as we will see, do hold under appropriate assumptions on $P_\theta$).
    In condition (v) we have also relaxed the assumption that the kernel is translation-invariant.
\end{remark}

\subsubsection{Verifying Regularity Conditions}
\label{app: ver reg cond}

\Cref{thm: vgd converge} applied to general loss functions $\mathcal{L}$; here we establish explicit sufficient conditions in the specific case of $\LPC$.

\begin{proposition}[Regularity for the variational gradient of $\LPC$] \label{lem: check reg cds}
    Let $p_\theta(y_i | x_i)$ be a positive density for all $\theta$ and each $(x_i,y_i)$ in the dataset.
    Let 
    \begin{enumerate}
        \item[(i)] $\sup_{\theta} p_\theta(y_i|x_i) < \infty$ \label{asm: part 1}
        \item[(ii)] $\displaystyle \sup_{\theta} \sqrt{k(\theta,\theta)} \frac{ \| \nabla_\theta p_\theta(y_i|x_i) \| }{p_\theta(y_i | x_i)} < \infty$ \label{asm: part 2}
        \item[(iii)] $\displaystyle \sup_{\theta} k(\theta,\theta) \frac{ | \Delta_\theta p_\theta(y_i|x_i) | }{ p_\theta(y_i | x_i) } < \infty$
        \item[(iv)] $\displaystyle \sup_{\theta} \| \nabla_1 k(\theta,\theta) \| \frac{ \| \nabla_\theta p_\theta(y_i|x_i) \| }{ p_\theta(y_i | x_i) } < \infty$
    \end{enumerate}
    for each $(x_i,y_i)$ in the dataset.
    Then the map $\theta \mapsto k(\theta,\theta) \|\nabla_\theta p_\theta(y_i | x_i)\|$ has at most linear growth and, for the loss function $\mathcal{L} = \LPC$ in \eqref{eq: define losses}, 
    \begin{align*}
        \sup_{Q \in \mathcal{P}(\mathbb{R}^d)} \left| \int k(\theta,\theta) \nabla \cdot \vargrad \mathcal{L}(Q)(\theta) \; \mathrm{d}Q(\theta) \right| & < \infty \\
        \sup_{Q \in \mathcal{P}(\mathbb{R}^d)} \left| \iint k(\theta,\vartheta) \vargrad \cdot \vargrad \mathcal{L}(Q)(\theta)(\vartheta) \; \mathrm{d}Q(\theta) \mathrm{d}Q(\vartheta) \right| & < \infty \\
        \sup_{Q \in \mathcal{P}(\mathbb{R}^d)} \left| \int \nabla_1 k(\theta,\theta) \cdot \vargrad \mathcal{L}(Q)(\theta) \; \mathrm{d}Q(\theta) \right| & < \infty .
    \end{align*}
\end{proposition}
\begin{proof}
    The first claim follows from combining (i) and (ii).
    Continuing from \Cref{prop: explicit variational grad}, the second variation is
    \begin{align*}
        \mathcal{L}''(Q)(\theta)(\vartheta) & = \sum_{i=1}^n \frac{p_\theta(y_i | x_i) p_\vartheta(y_i | x_i)}{p_Q(y_i | x_i)^2}, \nonumber
    \end{align*}
    which is well-defined since, under our assumptions, $p_Q(y_i | x_i) = \int p_\theta(y_i | x_i) \; \mathrm{d}Q(\theta)$ is strictly positive for all $Q \in \mathcal{P}(\mathbb{R}^d)$ and all $(x_i,y_i)$ in the dataset.
    The required variational gradients are
    \begin{align*}
        \nabla \cdot \vargrad \mathcal{L}(Q)(\theta) & = - \sum_{i=1}^n \frac{\Delta_\theta p_\theta(y_i | x_i)}{p_Q(y_i | x_i)} \\
        \vargrad \cdot \vargrad \mathcal{L}(Q)(\theta)(\vartheta) & = \sum_{i=1}^n \frac{ \nabla_\theta p_\theta(y_i | x_i) \cdot \nabla_\vartheta p_\vartheta(y_i | x_i) }{ p_Q(y_i | x_i)^2 } , 
    \end{align*}
    whose terms we assumed to be well-defined.
    Finally, since each $k(\theta,\theta) \Delta_\theta p_\theta(y_i | x_i)$ is bounded in $\theta$, $k(\theta,\theta) \nabla \cdot \vargrad \mathcal{L}(Q)$ is $Q$-integrable for each $Q \in \mathcal{P}(\mathbb{R}^d)$.
    Likewise, since $|k(\theta,\vartheta)| \leq \sqrt{k(\theta,\theta)} \sqrt{k(\vartheta,\vartheta)}$ and each $\sqrt{k(\theta,\theta)} \nabla_\theta p_\theta(y_i | x_i)$ is bounded in $\theta$, we deduce that $k(\theta,\vartheta) \vargrad \cdot \vargrad \mathcal{L}(Q)(\theta)(\vartheta)$ is $Q \otimes Q$-integrable for each $Q \in \mathcal{P}(\mathbb{R}^d)$.
    Integrating these equations and applying Jensen's inequality,
    \begin{align*}
        \left| \int k(\theta,\theta) \nabla \cdot \vargrad \mathcal{L}(Q)(\theta) \; \mathrm{d}Q(\theta) \right| & \leq \sum_{i=1}^n \frac{  \int k(\theta,\theta) | \Delta_\theta p_\theta(y_i | x_i) | \; \mathrm{d}Q(\theta) }{ \int p_\theta(y_i | x_i) \; \mathrm{d}Q(\theta) } \\
        \left| \iint k(\theta,\vartheta) \vargrad \cdot \vargrad \mathcal{L}(Q)(\theta)(\vartheta) \; \mathrm{d}Q(\theta) \mathrm{d}Q(\vartheta) \right| & \leq \sum_{i=1}^n \left( \frac{ \int \sqrt{k(\theta,\theta)} \| \nabla_\theta p_\theta(y_i | x_i) \| \; \mathrm{d}Q(\theta) }{ \int p_\theta(y_i | x_i) \; \mathrm{d}Q(\theta)  } \right)^2 .
    \end{align*}
    Let $\Pi_Q$ denote the distribution for which $(\mathrm{d}\Pi_Q / \mathrm{d}Q)(\theta) \propto p_\theta(y_i | x_i)$, so that
    \begin{align*}
        \frac{  \int k(\theta,\theta) | \Delta_\theta p_\theta(y_i | x_i) | \; \mathrm{d}Q(\theta) }{ \int p_\theta(y_i | x_i) \; \mathrm{d}Q(\theta) } & = \int \underbrace{ k(\theta,\theta) \frac{|\Delta_\theta p_\theta(y_i | x_i)|}{p_\theta(y_i | x_i)} }_{(*)} \; \mathrm{d}\Pi_Q(\theta) \\
        \frac{ \int \sqrt{k(\theta,\theta)} \| \nabla_\theta p_\theta(y_i | x_i) \| \; \mathrm{d}Q(\theta) }{ \int p_\theta(y_i | x_i) \; \mathrm{d}Q(\theta)  }  & = \int \underbrace{ \sqrt{k(\theta,\theta)} \frac{\|\nabla_\theta p_\theta(y_i | x_i)\|}{p_\theta(y_i | x_i)} }_{(**)} \; \mathrm{d}\Pi_Q(\theta) 
    \end{align*}
    where, under our assumptions, both integrands $(*)$ and $(**)$ are bounded over $\theta \in \mathbb{R}^d$.
    It follows that both integrals are bounded over $\Pi_Q \in \mathcal{P}(\mathbb{R}^d)$, and hence over $Q \in \mathcal{P}(\mathbb{R}^d)$.
    Similarly, from \eqref{eq: var grad Lmix},
    \begin{align*}
        \left| \int \nabla_1 k(\theta,\theta) \cdot \vargrad \mathcal{L}(Q)(\theta) \; \mathrm{d}Q(\theta) \right| \leq \sum_{i=1}^n \int \| \nabla_1 k(\theta,\theta) \| \frac{\| \nabla_\theta p_\theta(y_i | x_i) \|}{p_\theta(y_i | x_i)} \; \mathrm{d}\Pi_Q(\theta) ,
    \end{align*}
    where the integrand is bounded over $\theta \in \mathbb{R}^d$ by (iv), so that this integral is also bounded over $Q \in \mathcal{P}(\mathbb{R}^d)$, completing the argument.
\end{proof}

\subsubsection{Proof of \Cref{thm: main text}}
\label{app: proof main}

At last we can present a proof of \Cref{thm: main text}:

\begin{proof}[Proof of \Cref{thm: main text}]
Our first task is to verify the conditions of \Cref{thm: vgd converge} for $\mathcal{L} = \LPC$:
\begin{enumerate}
    \item[(i)] (Lower bound and chain rule) Since $\theta \mapsto p_\theta(y_i | x_i)$ is bounded for each $(x_i,y_i)$, $p_Q(y_i | x_i) \leq \sup_\theta p_\theta(y_i | x_i)$ for all $Q \in \mathcal{P}(\mathbb{R}^d)$, and thus
    \begin{align*}
        \mathcal{L}_{\min} \geq - \sum_{i=1}^n \log \sup_\theta p_\theta(y_i | x_i) > - \infty .
    \end{align*}
    Further, $\mathcal{L}(Q_N) = - \sum_{i=1}^n \log ( \frac{1}{N} \sum_{r=1}^N p_{\theta_r}(y_i | x_i) )$ is $C^2(\mathbb{R}^{d \times N})$, since each $\theta \mapsto p_\theta(y_i | x_i)$ is positive and $C^2(\mathbb{R}^d)$, and direct differentiation together with \eqref{eq: var grad Lmix} gives
    \begin{align*}
        \nabla_{\theta_j} \mathcal{L}(Q_N) = - \sum_{i=1}^n \frac{ \frac{1}{N} \nabla_{\theta_j} p_{\theta_j}(y_i | x_i) }{ \frac{1}{N} \sum_{r=1}^N p_{\theta_r}(y_i | x_i) } = \frac{1}{N} \vargrad \mathcal{L}(Q_N)(\theta_j) .
    \end{align*}
    Finally, by convexity of $z \mapsto - \log(z)$,
    \begin{align}
       &  \mathcal{L}(Q_N) \leq \frac{1}{N} \sum_{r=1}^N \left[ - \sum_{i=1}^n \log p_{\theta_r}(y_i | x_i) \right]  \nonumber \\
        & \quad \Rightarrow \quad 
        \mathbb{E}_{\theta_1,\dots,\theta_N \sim \mu_0}[\mathcal{L}(Q_N)] \leq - \sum_{i=1}^n \int \log p_\theta(y_i | x_i) \; \mathrm{d}\mu_0(\theta) , \label{eq: Jensen init}
    \end{align}
    which is finite since $\theta \mapsto p_\theta(y_i | x_i)$ is continuous and positive, and hence bounded away from zero on the bounded support of $\mu_0$.
    \item[(ii)] (Loss)  Since each $\theta \mapsto p_\theta(y_i | x_i)$ is positive and $C^2(\mathbb{R}^d)$, from \Cref{prop: explicit variational grad} the first variation $\mathcal{L}'(Q)$ is $C^2(\mathbb{R}^d)$ for each $Q \in \mathcal{P}(\mathbb{R}^d)$, and $\vargrad \cdot \vargrad \mathcal{L}(Q)$ is well-defined.
    Moreover, $(\theta_1 , \dots , \theta_N) \mapsto p_{Q_N}(y_i | x_i)$ is positive and $C^2(\mathbb{R}^{d \times N})$, so from \eqref{eq: var grad Lmix} the map $(\bm{\theta},\vartheta) \mapsto \vargrad \mathcal{L}(Q_N)(\vartheta) = - \sum_{i=1}^n \nabla_\vartheta p_\vartheta(y_i | x_i) / p_{Q_N}(y_i | x_i)$ is $C^1(\mathbb{R}^{d \times N} \times \mathbb{R}^d)$, and direct differentiation gives
    \begin{align*}
        \nabla_{\theta_r} \cdot [ \vargrad \mathcal{L}(Q_N)(\vartheta) ] & = \frac{1}{N} \sum_{i=1}^n \frac{ \nabla_\vartheta p_\vartheta(y_i | x_i) \cdot \nabla_{\theta_r} p_{\theta_r}(y_i | x_i) }{ p_{Q_N}(y_i | x_i)^2 } \\
        & = \frac{1}{N} \vargrad \cdot \vargrad \mathcal{L}(Q_N)(\vartheta)(\theta_r) ,
    \end{align*}
    which is the second chain rule required in (ii).
    Further, from \eqref{eq: var grad Lmix} and the fact that $\frac{1}{N} \sum_{r=1}^N p_{\theta_r}(y_i | x_i) \geq \frac{1}{N} p_{\theta_j}(y_i | x_i)$,
    \begin{align*}
        | \vargrad \mathcal{L}(Q_N)(\theta_j) | & = \left| - \sum_{i=1}^n \frac{\nabla_{\theta_j} p_{\theta_j}(y_i | x_i)  }{\frac{1}{N} \sum_{r=1}^N p_{\theta_r}(y_i | x_i)} 
        \right| \\
        & \leq N \sum_{i=1}^n \left| \frac{\nabla_{\theta_j} p_{\theta_j}(y_i|x_i)  }{ p_{\theta_j}(y_i | x_i) } \right|
        = N \sum_{i=1}^n \| \nabla_{\theta_j} \log p_{\theta_j}(y_i | x_i) \|  .
    \end{align*}
    From \Cref{lem: check reg cds} all three of the integrability conditions on the loss in (ii) of \Cref{thm: vgd converge} are satisfied.
    \item[(iii)] (Regularisation)  Satisfied by assumption.
    \item[(iv)] (Initialisation) Satisfied by assumption.
    \item[(v)] (Kernel) Satisfied by assumption.
    \item[(vi)] (Growth) Let $s := \max_i \sup_\theta \sqrt{k(\theta,\theta)} \| \nabla_\theta \log p_\theta(y_i | x_i) \|$, which is finite by condition (iv)a.\ of \Cref{thm: main text}.
    From the bound in (ii) and the bound $|k(\theta_i,\theta_j)| \leq \sqrt{k(\theta_i,\theta_i)} \sqrt{k(\theta_j,\theta_j)}$,
    \begin{align*}
        |k(\theta_i,\theta_j)| \| \vargrad \mathcal{L}(Q_N)(\theta_j) \| & \leq N \sum_{\ell=1}^n \sqrt{k(\theta_i,\theta_i)} \left[ \sqrt{k(\theta_j,\theta_j)} \| \nabla_{\theta_j} \log p_{\theta_j}(y_\ell | x_\ell) \| \right] \\
        & \leq N n s \sqrt{k(\theta_i,\theta_i)} .
    \end{align*}
    Since $k$ is symmetric, $\nabla_\theta [ k(\theta,\theta) ] = 2 \nabla_1 k(\theta,\theta)$, which has at most linear growth by condition (ii) of \Cref{thm: main text}; thus $k(\theta,\theta) \lesssim 1 + \|\theta\|^2$, and $\sqrt{k(\theta_i,\theta_i)}$, and hence $|k(\theta_i,\theta_j)| \| \vargrad \mathcal{L}(Q_N)(\theta_j) \|$, has at most linear growth.
    The remaining growth requirements were directly assumed.
\end{enumerate}

It remains only to establish that the constant $C_N$ appearing in the conclusion of \Cref{thm: vgd converge} can be bounded by an $N$-independent constant $C_0$.
Combining the lower bound on $\mathcal{L}_{\min}$ with \eqref{eq: Jensen init},
\begin{align*}
    C_N & = \KLD(\mu_0 || Q_0) + \mathbb{E}_{\theta_1,\dots,\theta_N \sim \mu_0}[\mathcal{L}(Q_N)] - \mathcal{L}_{\min} \\
    & \leq \KLD(\mu_0 || Q_0) + \sum_{i=1}^n \int \log \left( \frac{ \sup_\vartheta p_\vartheta(y_i | x_i) }{ p_\theta(y_i | x_i) } \right) \mathrm{d}\mu_0(\theta) =: C_0 ,
\end{align*}
which is finite, non-negative, and independent of both $N$ and $T$.
\color{black}
    This completes the argument.
\end{proof}

\subsection{Proof of Proposition \ref{prop: Gauss location MT}}
\label{app: gauss location}

This appendix is dedicated to a proof of our final theoretical result:

\begin{proof}[Proof of Proposition \ref{prop: Gauss location MT}]
    For these calculation we recall that $A : B = \mathrm{tr}(A B^\top)$ for matrices $A$, $B$ is the double dot product and that $[\nabla_\theta v(\theta)]_{i,j} = \nabla_{\theta_i} v_j(\theta)$ and $[\Delta_\theta v(\theta)]_{j} = \Delta_\theta v_j(\theta)$ for vector-valued $v : \mathbb{R}^d \rightarrow \mathbb{R}^d$.
    By convention $v(\theta)$ is a column vector and $\Delta_\theta v(\theta)$ is a row vector for $v : \mathbb{R}^d \rightarrow \mathbb{R}^d$.
    For the Gaussian location model
    \begin{align*}
    \nabla_\theta p_\theta(y_i | x_i) & =  (\nabla_\theta f_\theta(x_i)) \Sigma^{-1} (y_i - f_\theta(x_i)) p_\theta(y_i | x_i) \\
    \Delta_\theta p_\theta(y_i | x_i) & = (\Delta_\theta f_\theta(x_i)) \Sigma^{-1} (y_i - f_\theta(x_i)) p_\theta(y_i | x_i) \\
    & \qquad + (\nabla_\theta f_\theta(x_i)) : [ \nabla_\theta p_\theta(y_i | x_i) (y_i - f_\theta(x_i))^\top - (\nabla_\theta f_\theta(x_i)) p_\theta(y_i | x_i) ] \Sigma^{-1} \\
    & \hspace{-30pt} = \left\{ \begin{array}{l} (\Delta_\theta f_\theta(x_i)) \Sigma^{-1} (y_i - f_\theta(x_i)) \\
    \qquad + (\nabla_\theta f_\theta(x_i)) : [ (\nabla_\theta f_\theta(x_i)) \Sigma^{-1} (y_i -f_\theta(x_i))(y_i -f_\theta(x_i))^\top \Sigma^{-1} ] \\
    \qquad - (\nabla_\theta f_\theta(x_i)) : [ (\nabla_\theta f_\theta(x_i)) \Sigma^{-1} ]
    \end{array} \right\} p_\theta(y_i | x_i)
    \end{align*}
    so that
    \begin{align*}
        \sup_{\theta} \sqrt{k(\theta,\theta)} \frac{ \| \nabla_\theta p_\theta(y_i|x_i) \| }{p_\theta(y_i | x_i)} & = \sup_{\theta} \sqrt{k(\theta,\theta)} \| (\nabla_\theta f_\theta(x_i)) \Sigma^{-1} (y_i - f_\theta(x_i)) \| \\
        & \leq \|\Sigma^{-1} y_i\| \sup_{\theta} \sqrt{k(\theta,\theta)} \| (\nabla_\theta f_\theta(x_i)) \|_{\mathrm{op}} \\
        & \qquad + \| \Sigma^{-1} \|_{\mathrm{op}} \left[ \sup_{\theta } \|f_\theta(x_i)\| \right] \left[ \sup_{\theta} \sqrt{k(\theta,\theta)} \| (\nabla_\theta f_\theta(x_i)) \|_{\mathrm{op}} \right]
    \end{align*}
    where the finiteness of the terms appearing on the right hand side was assumed.
    An identical argument, with $\sqrt{k(\theta,\theta)}$ replaced by $\|\nabla_1 k(\theta,\theta)\|$, shows that condition (iv)c.\ of \Cref{thm: main text} is satisfied.
    A similar but more lengthy calculation (which we omit for brevity) for the supremum in condition (iv)b. completes the argument.
\end{proof}

\subsection{Proof of \Cref{thm: bootstrap} {and \Cref{cor: type I}}}
\label{app: bootstrap proof}

{
This appendix is devoted to the proofs of \Cref{thm: bootstrap} and \Cref{cor: type I}.
First we present three auxiliary results; the equivariance of the posteriors and the test statistic for linear models (\Cref{prop: equivariance}), the stability of entropy-regularised minimisers with respect to perturbations of the reference distribution (\Cref{prop: prior stab}), and control of the \ac{MMD} between mixtures using total variation (\Cref{prop: MMD TV}).
Then we present the proofs of \Cref{thm: bootstrap} and \Cref{cor: type I}.

For $Q_0 \in \mathcal{P}(\mathbb{R}^d)$ and a dataset $\mathfrak{D} = \{(x_i,y_i)\}_{i=1}^n$, let $\QBayes^{\mathfrak{D},Q_0}$ and $\QPC^{\mathfrak{D},Q_0}$ denote the minimisers of \eqref{eq: objective} with $\mathcal{L} = \mathcal{L}_{\mathrm{Bayes}}^{\mathfrak{D}}$ and $\mathcal{L} = \mathcal{L}_{\mathrm{PrO}}^{\mathfrak{D}}$ respectively, where the superscripts make the dependence on the dataset and the prior explicit, and let $\PBayes^{\mathfrak{D},Q_0}$ and $\PPC^{\mathfrak{D},Q_0}$ denote the associated predictive distributions.
Under assumptions (i) and (ii) of \Cref{thm: bootstrap}, these minimisers exist and are unique, by \Cref{prop: Qmix well def} for $\QPC$ and by the variational characterisation of the Bayesian posterior \citep{zellner1988optimal} for $\QBayes$.

For $c \in \mathbb{R}^p$ let $s_c(y) := y + c$, and for $\Delta \in \mathbb{R}^d$ let $\tau_\Delta(\theta) := \theta - \Delta$.

\begin{proposition}[Equivariance for Linear Models]
\label{prop: equivariance}
Assume that condition (i) of \Cref{thm: bootstrap} holds.
Let $\mathfrak{D} = \{(x_i,y_i)\}_{i=1}^n$, let $\Delta \in \mathbb{R}^d$, and let $\mathfrak{D}' = \{(x_i, y_i + A(x_i) \Delta)\}_{i=1}^n$.
Then, for $\mathrm{M} \in \{\mathrm{Bayes},\mathrm{PrO}\}$, 
\begin{align*}
Q_{\mathrm{M}}^{\mathfrak{D}',Q_0} = (\tau_\Delta^{-1})_\# Q_{\mathrm{M}}^{\mathfrak{D},(\tau_\Delta)_\# Q_0} \qquad \text{and} \qquad P_{\mathrm{M}}^{\mathfrak{D}',Q_0}(\cdot|x) = (s_{A(x) \Delta})_\# P_{\mathrm{M}}^{\mathfrak{D},(\tau_\Delta)_\# Q_0}(\cdot | x)
\end{align*}
for all $x$.
\end{proposition}
\begin{proof}
From condition (i), $p_{\vartheta}(y + A(x)\Delta | x) = g_x(y - A(x)(\vartheta - \Delta)) = p_{\tau_\Delta(\vartheta)}(y|x)$ for all $\vartheta$, $x$ and $y$.
Thus, for any $Q \in \mathcal{P}(\mathbb{R}^d)$,
\begin{align*}
\mathcal{L}_{\mathrm{PrO}}^{\mathfrak{D}'}(Q) & = - \sum_{i=1}^n \log \int p_{\tau_\Delta(\vartheta)}(y_i|x_i) \, \mathrm{d}Q(\vartheta) \\
& = - \sum_{i=1}^n \log \int p_{\theta}(y_i|x_i) \, \mathrm{d}((\tau_\Delta)_\# Q)(\theta) = \mathcal{L}_{\mathrm{PrO}}^{\mathfrak{D}}((\tau_\Delta)_\# Q) ,
\end{align*}
and an identical argument shows that $\mathcal{L}_{\mathrm{Bayes}}^{\mathfrak{D}'}(Q) = \mathcal{L}_{\mathrm{Bayes}}^{\mathfrak{D}}((\tau_\Delta)_\# Q)$.
Since the \ac{KLD} is invariant under the application of a common bijection to both of its arguments, $\KLD(Q || Q_0) = \KLD((\tau_\Delta)_\# Q || (\tau_\Delta)_\# Q_0)$.
It follows that the objective \eqref{eq: objective} based on $\mathfrak{D}'$ and $Q_0$, evaluated at $Q$, is equal to the objective based on $\mathfrak{D}$ and $(\tau_\Delta)_\# Q_0$, evaluated at $(\tau_\Delta)_\# Q$.
Since $Q \mapsto (\tau_\Delta)_\# Q$ is a bijection on $\mathcal{P}(\mathbb{R}^d)$ and the minimisers are unique, the first claim is established.
For the second claim, note from condition (i) that $P_{\vartheta}(\cdot|x) = (s_{A(x)\Delta})_\# P_{\tau_\Delta(\vartheta)}(\cdot|x)$, and thus, using the first claim,
\begin{align*}
P_{\mathrm{M}}^{\mathfrak{D}',Q_0}(\cdot|x) = \int P_\vartheta(\cdot|x) \, \mathrm{d}Q_{\mathrm{M}}^{\mathfrak{D}',Q_0}(\vartheta) = (s_{A(x)\Delta})_\# \int P_\theta(\cdot|x) \, \mathrm{d}Q_{\mathrm{M}}^{\mathfrak{D},(\tau_\Delta)_\# Q_0}(\theta) ,
\end{align*}
as claimed.
\end{proof}

\begin{proposition}[Stability with Respect to the Prior]
\label{prop: prior stab}
Let $\mathcal{L} : \mathcal{P}(\mathbb{R}^d) \rightarrow (-\infty,\infty]$ be convex and bounded below.
Let $\pi_1, \pi_2 \in \mathcal{P}(\mathbb{R}^d)$ be mutually absolutely continuous with $h := \log(\mathrm{d}\pi_2 / \mathrm{d}\pi_1)$ satisfying $\sup_\theta |h(\theta)| \leq \delta$.
For $k \in \{1,2\}$ suppose that $\mathcal{J}_k(Q) := \mathcal{L}(Q) + \KLD(Q || \pi_k)$ has a minimiser $R_k$ with $\mathcal{J}_k(R_k) < \infty$.
Then $\|R_1 - R_2\|_{\mathrm{TV}} \leq 2 \delta$.
\end{proposition}
\begin{proof}
First we show that $\mathcal{J}_k(Q) - \mathcal{J}_k(R_k) \geq \KLD(Q || R_k)$ for all $Q$ with $\mathcal{J}_k(Q) < \infty$, i.e. that $\mathcal{J}_k$ is $1$-strongly convex with respect to the \ac{KLD} at its minimiser.
Fix such a $Q$, and for $t \in (0,1)$ let $R_t := (1-t) R_k + t Q$.
Since $\mathcal{L}$ is bounded below, $\KLD(R_k || \pi_k)$ and $\KLD(Q || \pi_k)$ are finite, and direct calculation with densities with respect to $\pi_k$ gives the identity
\begin{align*}
(1-t) \KLD(R_k || \pi_k) + t \KLD(Q || \pi_k) = \KLD(R_t || \pi_k) + (1-t) \KLD(R_k || R_t) + t \KLD(Q || R_t) .
\end{align*}
Combining this identity with the convexity of $\mathcal{L}$ and the optimality of $R_k$,
\begin{align*}
\mathcal{J}_k(R_k) \leq \mathcal{J}_k(R_t) \leq (1-t) \mathcal{J}_k(R_k) + t \mathcal{J}_k(Q) - (1-t) \KLD(R_k || R_t) - t \KLD(Q || R_t) ,
\end{align*}
and rearranging gives $\mathcal{J}_k(Q) - \mathcal{J}_k(R_k) \geq \KLD(Q || R_t)$.
Since $R_t \rightarrow R_k$ in total variation as $t \rightarrow 0$, the lower semicontinuity of the \ac{KLD} gives $\mathcal{J}_k(Q) - \mathcal{J}_k(R_k) \geq \KLD(Q || R_k)$, as claimed.

Now, since $\KLD(Q||\pi_1) = \KLD(Q || \pi_2) + \int h \, \mathrm{d}Q$ whenever either side is finite, and $h$ is bounded, we have $\mathcal{J}_1(Q) = \mathcal{J}_2(Q) + \int h \, \mathrm{d}Q$, and in particular $\mathcal{J}_1(R_2) < \infty$ and $\mathcal{J}_2(R_1) < \infty$.
Applying the first part of the proof twice,
\begin{align*}
\KLD(R_2 || R_1) + \KLD(R_1 || R_2) & \leq \{ \mathcal{J}_1(R_2) - \mathcal{J}_1(R_1) \} + \{ \mathcal{J}_2(R_1) - \mathcal{J}_2(R_2) \} \\
& = \int h \, \mathrm{d}(R_2 - R_1) \leq \delta \|R_1 - R_2\|_{\mathrm{TV}} \leq \delta \sqrt{2 \KLD(R_2 || R_1)} ,
\end{align*}
where the last step used Pinsker's inequality.
Thus $\KLD(R_2 || R_1) \leq 2 \delta^2$, and a final application of Pinsker's inequality gives $\|R_1 - R_2\|_{\mathrm{TV}} \leq \sqrt{2 \KLD(R_2 || R_1)} \leq 2 \delta$.
\end{proof}

\begin{proposition}[Controlling MMD by Total Variation]
\label{prop: MMD TV}
Assume that condition (iii) of \Cref{thm: bootstrap} holds.
Then, for all $x$ and all $R, R' \in \mathcal{P}(\mathbb{R}^d)$,
\begin{align*}
\mathrm{MMD}_\kappa\left( \int P_\theta(\cdot|x) \, \mathrm{d}R(\theta) , \int P_\theta(\cdot|x) \, \mathrm{d}R'(\theta) \right) \leq M \|R - R'\|_{\mathrm{TV}} .
\end{align*}
Moreover, $\mathrm{MMD}_\kappa(s_c{}_\# P , s_c{}_\# P') = \mathrm{MMD}_\kappa(P,P')$ for all $c \in \mathbb{R}^p$ and all distributions $P$, $P'$.
\end{proposition}
\begin{proof}
For any distribution $P$, $\|\mu_P\|_{\mathcal{H}_\kappa}^2 = \iint \kappa(y,y') \, \mathrm{d}P(y) \mathrm{d}P(y') \leq M^2$, since $|\kappa(y,y')| \leq \sqrt{\kappa(y,y) \kappa(y',y')} \leq M^2$.
The kernel mean embedding of the mixture $\int P_\theta(\cdot|x) \, \mathrm{d}R(\theta)$ is the Bochner integral $\int \mu_{P_\theta(\cdot|x)} \, \mathrm{d}R(\theta)$, and thus
\begin{align*}
\mathrm{MMD}_\kappa\left( \int P_\theta(\cdot|x) \, \mathrm{d}R(\theta) , \int P_\theta(\cdot|x) \, \mathrm{d}R'(\theta) \right) & = \left\| \int \mu_{P_\theta(\cdot|x)} \, \mathrm{d}(R - R')(\theta) \right\|_{\mathcal{H}_\kappa} \\
& \leq M \, \|R - R'\|_{\mathrm{TV}} ,
\end{align*}
which establishes the first claim.
For the second claim, translation-invariance of $\kappa$ gives
\begin{align*}
\mathrm{MMD}_\kappa^2(s_c{}_\# P , s_c{}_\# P') = \iint \kappa(y + c, y' + c) \, \mathrm{d}(P - P')(y) \mathrm{d}(P - P')(y') = \mathrm{MMD}_\kappa^2(P,P') ,
\end{align*}
as claimed.
\end{proof}
}

\subsubsection{Proof of \Cref{thm: bootstrap}}
{

\begin{proof}[Proof of \Cref{thm: bootstrap}]
Fix $\theta$, $\vartheta$, $u$, and the covariates, and let $\Delta := \vartheta - \theta$.
Let $\mathfrak{D}$ and $\mathfrak{D}'$ denote the datasets generated by \eqref{eq: data generation} using, respectively, $(\theta,u)$ and $(\vartheta,u)$.
From condition (i), $G(\vartheta,x_i,u) = G(\theta,x_i,u) + A(x_i)\Delta$, so that $\mathfrak{D}$ and $\mathfrak{D}'$ are related as in \Cref{prop: equivariance}.
Let $\tilde{Q}_0 := (\tau_\Delta)_\# Q_0$.
From \Cref{prop: equivariance}, the predictive distributions based on $\mathfrak{D}'$ and $Q_0$ are obtained by translating the predictive distributions based on $\mathfrak{D}$ and $\tilde{Q}_0$ by $A(x_i)\Delta$, and this translation is common to the Bayesian and \ac{PC} predictive distributions.
Thus, from the second part of \Cref{prop: MMD TV},
\begin{align*}
\mathcal{T}_n(\vartheta,u) & = \frac{1}{n} \sum_{i=1}^n \mathrm{MMD}_\kappa^2\left( \PPC^{\mathfrak{D},\tilde{Q}_0}(\cdot|x_i) , \PBayes^{\mathfrak{D},\tilde{Q}_0}(\cdot|x_i) \right) \\ 
\mathcal{T}_n(\theta,u) & = \frac{1}{n} \sum_{i=1}^n \mathrm{MMD}_\kappa^2\left( \PPC^{\mathfrak{D},Q_0}(\cdot|x_i) , \PBayes^{\mathfrak{D},Q_0}(\cdot|x_i) \right) ;
\end{align*}
that is, the two test statistics are computed from the same dataset, and differ only in the location of the prior.
For a prior $\pi$, let 
$$
d_i(\pi) := \mu_{\PPC^{\mathfrak{D},\pi}(\cdot|x_i)} - \mu_{\PBayes^{\mathfrak{D},\pi}(\cdot|x_i)} \in \mathcal{H}_\kappa .
$$
Then $\sqrt{\mathcal{T}_n(\theta,u)}$ is the norm of $(d_1(Q_0),\dots,d_n(Q_0))$ in the Hilbert space $\mathcal{H}_\kappa^n$ equipped with the inner product $\langle (a_i) , (b_i) \rangle := \frac{1}{n} \sum_{i=1}^n \langle a_i , b_i \rangle_{\mathcal{H}_\kappa}$, and similarly $\sqrt{\mathcal{T}_n(\vartheta,u)}$ is the norm of $(d_1(\tilde{Q}_0),\dots,d_n(\tilde{Q}_0))$.
From the reverse triangle inequality in this space, the triangle inequality in $\mathcal{H}_\kappa$, and the first part of \Cref{prop: MMD TV},
\begin{align*}
\left| \sqrt{\mathcal{T}_n(\vartheta,u)} - \sqrt{\mathcal{T}_n(\theta,u)} \right| & \leq \max_{i = 1,\dots,n} \|d_i(\tilde{Q}_0) - d_i(Q_0)\|_{\mathcal{H}_\kappa} \\
& \leq M \left\{ \| \QPC^{\mathfrak{D},\tilde{Q}_0} - \QPC^{\mathfrak{D},Q_0} \|_{\mathrm{TV}} + \| \QBayes^{\mathfrak{D},\tilde{Q}_0} - \QBayes^{\mathfrak{D},Q_0} \|_{\mathrm{TV}} \right\} .
\end{align*}
Finally we apply \Cref{prop: prior stab} with $\pi_1 = Q_0$ and $\pi_2 = \tilde{Q}_0$, for which $h(\theta) = \log q_0(\theta + \Delta) - \log q_0(\theta)$ satisfies $\sup_\theta |h(\theta)| \leq L_0 \|\Delta\|$ from condition (ii).
The losses $\LPC$ and $\LBayes$ are both convex, and bounded below by $- \sum_{i=1}^n \log \sup_y g_{x_i}(y) > - \infty$ from condition (i).
Moreover, the minimal values of the objectives are finite for any prior $\pi$ with a positive density; for $\LPC$ the objective evaluated at $\pi$ itself is $- \sum_{i=1}^n \log \int p_\theta(y_i|x_i) \, \mathrm{d}\pi(\theta)$, which is finite since the integrands are positive and bounded, while for $\LBayes$ the minimal value is $- \log \int \prod_{i=1}^n p_\theta(y_i|x_i) \, \mathrm{d}\pi(\theta)$, which is finite for the same reason.
Thus \Cref{prop: prior stab} applies to both posteriors, giving
\begin{align*}
\left| \sqrt{\mathcal{T}_n(\vartheta,u)} - \sqrt{\mathcal{T}_n(\theta,u)} \right| \leq M \left\{ 2 L_0 \|\Delta\| + 2 L_0 \|\Delta\| \right\} = 4 M L_0 \|\vartheta - \theta\| ,
\end{align*}
which establishes the first claim.

For the second claim, fix $\theta$ and $\vartheta$, and let $\epsilon := K \|\vartheta - \theta\|$.
Let $F_\theta$ denote the distribution function of $\sqrt{\mathcal{T}_n(\theta,U)}$, $U \sim \nu$, and let $q_\theta := \inf\{t : F_\theta(t) \geq 1 - \alpha\}$ denote its $(1-\alpha)$-quantile; similarly for $\vartheta$.
Since $\sqrt{\mathcal{T}_n(\vartheta,u)} \leq \sqrt{\mathcal{T}_n(\theta,u)} + \epsilon$ for all $u$, we have $F_\vartheta(t + \epsilon) \geq F_\theta(t)$ for all $t$, and setting $t = q_\theta$ and using right-continuity of $F_\theta$ gives $F_\vartheta(q_\theta + \epsilon) \geq 1 - \alpha$, so that $q_\vartheta \leq q_\theta + \epsilon$.
Exchanging the roles of $\theta$ and $\vartheta$ gives $|q_\vartheta - q_\theta| \leq \epsilon$, as claimed.
\end{proof}
}

\subsubsection{Proof of \Cref{cor: type I}}
{

\begin{proof}[Proof of \Cref{cor: type I}]
Let $q_n(\theta)$ denote the $(1-\alpha)$-quantile of $\sqrt{\mathcal{T}_n(\theta,U)}$, $U \sim \nu$, so that the bootstrap critical value (in the idealised setting where infinitely many bootstrap replicates are used) is $\hat{c}_n = q_n(\theta_n)^2$, and the test rejects if and only if $\sqrt{r_n \mathcal{T}_n(\theta_\star,u)} > \sqrt{r_n} \, q_n(\theta_n)$.
From \Cref{thm: bootstrap},
\begin{align*}
\left| \sqrt{r_n} \, q_n(\theta_n) - \sqrt{r_n} \, q_n(\theta_\star) \right| \leq K \sqrt{r_n} \|\theta_n - \theta_\star\| \stackrel{p}{\rightarrow} 0 .
\end{align*}
Let $Z \sim F$ and let $c_\alpha$ denote the $(1-\alpha)$-quantile of $F$.
From the continuous mapping theorem, $\sqrt{r_n \mathcal{T}_n(\theta_\star,U)} \stackrel{d}{\rightarrow} \sqrt{Z}$, and the distribution function $s \mapsto F(s^2)$ of $\sqrt{Z}$ is continuous and strictly increasing at its $(1-\alpha)$-quantile $\sqrt{c_\alpha}$.
Since $\sqrt{r_n} \, q_n(\theta_\star)$ is the $(1-\alpha)$-quantile of $\sqrt{r_n \mathcal{T}_n(\theta_\star,U)}$, convergence of quantiles \citep[Lemma 21.2 of][]{vandervaart1998asymptotic} gives $\sqrt{r_n} \, q_n(\theta_\star) \rightarrow \sqrt{c_\alpha}$, and hence $\sqrt{r_n} \, q_n(\theta_n) \stackrel{p}{\rightarrow} \sqrt{c_\alpha}$.
Under the null, the observed test statistic is $\mathcal{T}_n(\theta_\star,u)$ with $u \sim \nu$, and $\theta_n$ is computed from the same data.
Although these two quantities are dependent, $\sqrt{r_n} \, q_n(\theta_n)$ converges in probability to a constant, so that from Slutsky's theorem
\begin{align*}
\sqrt{r_n \mathcal{T}_n(\theta_\star,u)} - \sqrt{r_n} \, q_n(\theta_n) \stackrel{d}{\rightarrow} \sqrt{Z} - \sqrt{c_\alpha} .
\end{align*}
Since $F$ is continuous at $c_\alpha$, $\mathbb{P}(\sqrt{Z} = \sqrt{c_\alpha}) = 0$, and the portmanteau theorem gives
\begin{align*}
\mathbb{P}\left( \mathcal{T}_n(\theta_\star,u) > \hat{c}_n \right) \rightarrow \mathbb{P}(Z > c_\alpha) = 1 - F(c_\alpha) = \alpha ,
\end{align*}
as claimed.
\end{proof}
}

\color{black}

\section{Experimental Protocol}
\label{app: sims}

This appendix contains the details required to reproduce the experiments reported in \Cref{sec: empirical}.
The test problems that we consider are specified in \Cref{app: test problems}, details of the maximum mean discrepancy test statistic are contained in \Cref{app: mmd}, implementational aspects of \ac{VGD} are discussed in \Cref{app: vgd}, and additional empirical results are contained in \Cref{app: additional}.
Full details for the seismic travel time tomography case study are contained in \Cref{app: seismic}.

\paragraph{Code} 

Code to reproduce our simulation study from \Cref{sec: sims} is available at \url{https://github.com/liuqingyang27/Detecting-Model-Misspecification-via-VGD}.
Code to reproduce our seismic travel time tomography experiments from \Cref{sec: waveform} is available at \url{https://github.com/XuebinZhaoZXB/Detecting-Model-Misspecification/}.

\subsection{Test Problems}
\label{app: test problems}

The regression functions that we considered for our simulation study in \Cref{sec: sims} are as follows:
\begin{enumerate}
    \item $f_\theta(x) = \theta x^2$ with $\theta \in \mathbb{R}$ and $\{x_i\}$ uniformly sampled from $[0,1]$
    \item $f_\theta(x) = \frac{1}{1+e^{-\theta x}}$ with $\theta \in \mathbb{R}$ and $\{x_i\}$ uniformly sampled from $[-1,1]$
    \item $f_\theta(x) = \theta_1 + \theta_2 x$ with $\theta \in \mathbb{R}^2$ and $\{x_i\}$ uniformly sampled from $[-2,2]$
\end{enumerate}
In the well-specified scenario, data are generated by $y_i=f_\theta(x_i)+z_i$ for all $i\in \{1,\cdots, n\}$, where $z_i$ are i.i.d $\mathcal{N}(0,\sigma^2)$. For the three cases above, the noise levels are separately $0.5, 0.05$ and $0.8$.
The true data-parameters parameters in this case were $\theta = 5$ for the quadratic model, $\theta = 5$ for the sigmoid model and $\theta = [5, 3]$ for the linear model.
To generate data that are misspecified we proceeded as follows for each of the above tasks:
\begin{enumerate}
    \item $y_i = (5+3 u_i)x_i^2 +z_i$ for all $i\in \{1,\cdots, n\}$, where $u_i\sim \mathcal{N}(0,1),\ z_i\sim \mathcal{N}(0,0.5^2)$ 
    \item data points are generated from a uniform distribution with density
    \begin{equation*}
      f(x, y) =
      \begin{cases}
        \frac{1}{2}, & \text{if } (x, y) \in (0, 1) \times (0, 1), \\
        \frac{1}{2}, & \text{if } (x, y) \in (-1, 0) \times (-1, 0), \\
        0, & \text{otherwise}.
      \end{cases}
    \end{equation*}    
    \item $y_i = 5 + 3 x_i + 2 x_i^2 + z_i$ for all $i\in \{1,\cdots, n\}$, where $z_i\sim \mathcal{N}(0,0.8^2).$
\end{enumerate}

In the second of the above examples $\theta \mapsto f_\theta(x)$, $\nabla_\theta f_\theta(x)$ and $\Delta_\theta f_\theta(x)$ are bounded, so from Proposition \ref{prop: Gauss location MT} the sufficient conditions of our theory are satisfied whenever the kernel $k$ is bounded.
On the other hand, in the first and third cases our theoretical assumptions are \emph{not} satisfied; this enables us to test the performance of \Cref{alg: svgd} outside the setting where our theoretical results hold.

\subsection{Maximum Mean Discrepancy}
\label{app: mmd}

The maximum mean discrepancy employed in these experiments utilised a Gaussian kernel 
\begin{align*}
    \kappa(y,y') = \exp \left( - \frac{\|y-y'\|^2}{2 \ell^2} \right)
\end{align*}
where the lengthscale was selected as the standard deviation of $\{y_i\}_{i=1}^n$.
For the synthetic examples we present in \Cref{sec: sims}, the dimension of the response variable is always $p = 1$, but for completeness here we work with the general form of the Gaussian kernel.
The choice of a Gaussian kernel together with the Gaussian measurement error model \eqref{eq: Gaussian location} with covariance matrix $\Sigma = \sigma^2 I_{p \times p}$ enables \eqref{eq: MMD} to be explicitly computed using the analytic form of the integral
\begin{align*}
    \mathfrak{k}(\theta,\vartheta | x_i) & := \iint \exp \left( - \frac{\|y - y' \|^2}{2 \ell^2} \right) \; \mathrm{d}P_\theta(y|x_i) \mathrm{d}P_\vartheta(y'|x_i) \\
    & = \left( \frac{\ell^2}{\ell^2 + 2 \sigma^2} \right)^{p/2} \exp\left( - \frac{ \| f_\theta(x_i) - f_\vartheta(x_i) \|^2 }{2(\ell^2 + 2\sigma^2)} \right) 
\end{align*}
together with
\begin{align}
    \mathrm{MMD}^2(\PBayes(\cdot | x_i),\PPC(\cdot | x_i)) & = \iint \mathfrak{k}(\theta,\vartheta | x_i) \; \mathrm{d}\QBayes(\theta) \mathrm{d}\QBayes(\vartheta) \nonumber \\
    & \qquad - 2 \iint \mathfrak{k}(\theta,\vartheta | x_i) \; \mathrm{d}\QBayes(\theta) \mathrm{d}\QPC(\vartheta) \nonumber \\
    & \qquad + \iint \mathfrak{k}(\theta,\vartheta | x_i) \; \mathrm{d}\QPC(\theta) \mathrm{d}\QPC(\vartheta). \label{eq: explicit MMD}
\end{align}
In practice both $\QBayes$ and $\QPC$ are approximated using \ac{VGD}, so we obtain a particle-based representation $\{\theta_i^{\mathrm{Bayes}}\}_{i=1}^N$ for $\QBayes$ and $\{\theta_i^{\mathrm{PrO}}\}_{i=1}^N$ for $\QPC$.
Substituting these empirical measures into \eqref{eq: explicit MMD} we obtain
\begin{align}
    \mathrm{MMD}^2(\PBayes(\cdot | x_i),\PPC(\cdot | x_i)) & \approx \frac{1}{N^2} \sum_{r=1}^N \sum_{s=1}^N \mathfrak{k}(\theta_r^{\mathrm{Bayes}},\theta_s^{\mathrm{Bayes}} | x_i) \nonumber  \\
    & \qquad - 2 \frac{1}{N^2} \sum_{r=1}^N \sum_{s=1}^N \mathfrak{k}(\theta_r^{\mathrm{Bayes}},\theta_s^{\mathrm{PrO}} | x_i)  \nonumber \\
    & \qquad + \frac{1}{N^2} \sum_{r=1}^N \sum_{s=1}^N \mathfrak{k}(\theta_r^{\mathrm{PrO}},\theta_s^{\mathrm{PrO}} | x_i) . \label{eq: MMD empirical}
\end{align}
The approximate values in \eqref{eq: MMD empirical} were used for the experiments that we report in \Cref{sec: empirical} of the main text.

\subsection{Variational Gradient Descent}
\label{app: vgd}

For our toy experiments we utilised the inverse multiquadric kernel
$$
k(\theta,\vartheta) = \left(c^2 + \frac{\|\theta - \vartheta\|^2}{l^2}\right)^{-\beta}
$$
with $\beta = 0.5$.
To select an appropriate length scale $l$, we employed the median heuristic \citep{garreau2017large} at each iteration of \Cref{alg: svgd}.
The step size and iteration number in each experiment were manually selected to ensure convergence, as quantified by \ac{KGD} (see e.g. \Cref{fig: sim study 2}). 
In all experiments $N=20$ particles were used.

\subsection{Additional Empirical Results}
\label{app: additional}

The posterior distributions $\QBayes$ and $\QPC$ corresponding to the regression tasks in \Cref{fig: sim study} are displayed in \Cref{fig: sim study 2}, alongside the values of the \ac{KGD} in \eqref{eq: KGD} obtained along the trajectory of \ac{VGD}.
A kernel density estimator has been applied to the particle representations of $\QBayes$ and $\QPC$ to aid visualisation in \Cref{fig: sim study 2}.
It can be seen that the standard Bayesian posterior $\QBayes$ is rather concentrated in all scenarios, irrespective of whether the statistical model is well-specified or misspecified, while $\QPC$ tends to be more diffuse when the statistical model is misspecified.
The values of \ac{KGD} obtained along the trajectory of \ac{VGD} appear to generally decrease and converge to a limit in all cases, consistent with an accurate $N$-particle approximation having been found.

\begin{figure}[t!]
    \includegraphics[width=\textwidth]{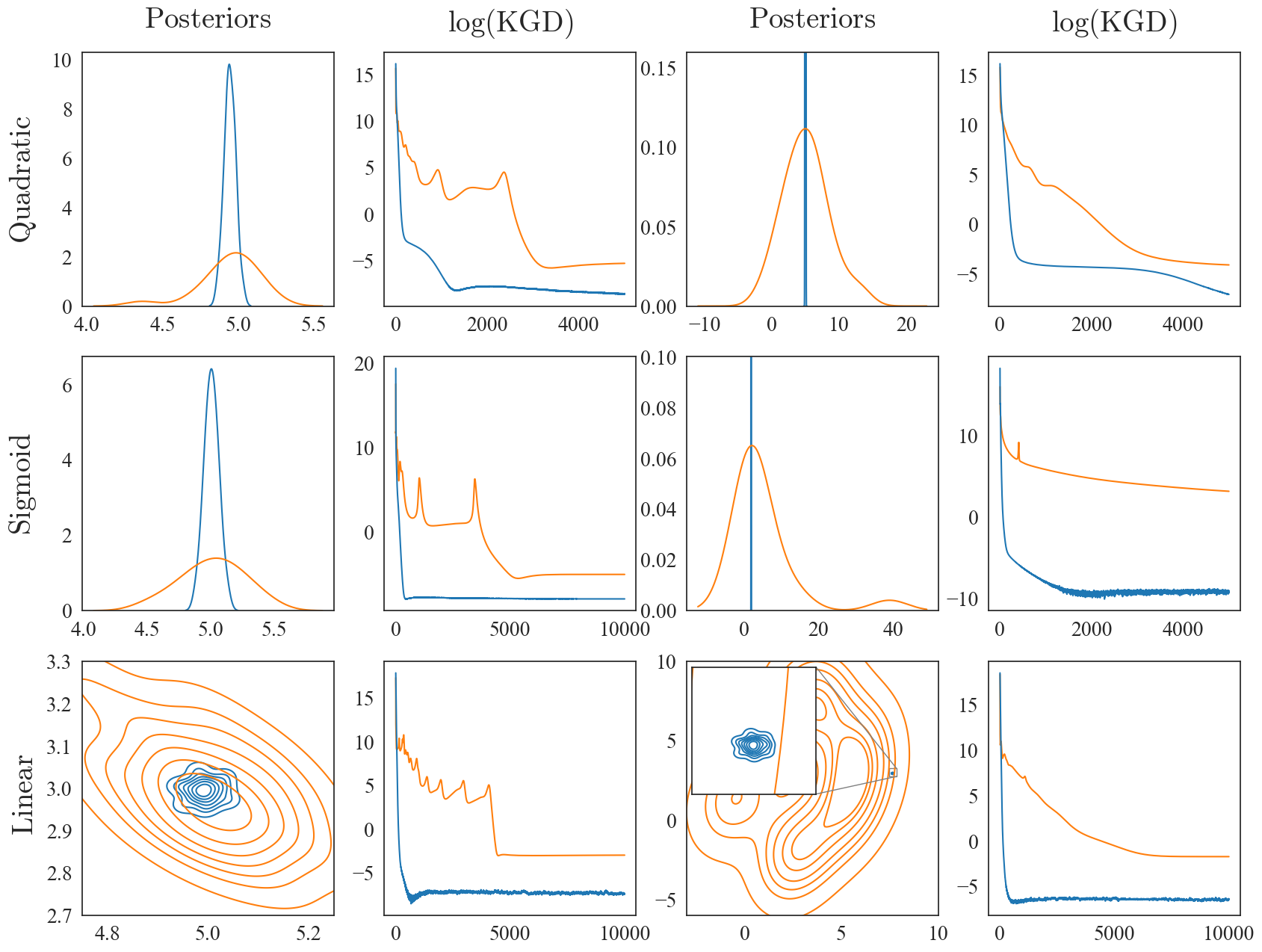}
    \caption{Simulation Study.  Each row corresponds to a regression task in \Cref{fig: sim study} in which the data are either generated from the statistical model (well-specified, left) or not generated from the statistical model (misspecified, right). The posterior distributions $\QBayes$ (blue) and $\QPC$ (orange) are displayed, together with the values of the \ac{KGD} in \eqref{eq: KGD} obtained along the trajectory of \ac{VGD}.
    \alttext{A figure consisting of 12 panels arranged into a $3 \times 4$ grid.  Each row represents a different dataset, with both the parameter posteriors and the (log) KGD values displayed for settings in which the statistical model is well-specified and misspecified.}
    }
    \label{fig: sim study 2}
\end{figure}

Further to the analysis in the main text, we investigate the performance of our method under different data sizes $n$ and different dimensions of $\theta$. 

First, in \Cref{fig: diff data size}, we considered the sigmoid regression task with varying sample sizes of $n = 100, 1000$ and $10000$. 
The test statistic values $\mathcal{T}(\{(x_i,\tilde{y}_i)\}_{i=1}^n)$ calculated under the (bootstrap) null typically decrease as the data size grows, while under misspecification the actual $\mathcal{T}$ values are effectively $n$-independent. 
Consequently, misspecified models are easier to detect when we have a larger dataset, as would be expected.

Second, in \Cref{fig: diff dimension} we consider a regression model defined by 
\begin{equation}
f_\theta(x) = \sum^{d}_{p=1} \theta_p \sin(px).\label{eq: diff dim model}    
\end{equation}
For the misspecified scenario, the data are generated according to $y_i = \sin(1/x_i) + z_i$ where $z_i \sim \mathcal{N}(0,\sigma^2)$ and $\sigma = 0.2$, a function that remains misspecified regardless of the dimension $p$ of the model \eqref{eq: diff dim model}.
In particular, we cannot expect \eqref{eq: diff dim model} to resolve the rapid oscillations around $x = 0$.
For presentational purposes we consider $p \in \{5,20,50\}$.
The predictive distributions $\PBayes$ and $\PPC$ perform generally well when the model is well-specified.
In the misspecified case, $\PBayes$ is over-confident around $x = 0$ and this is partially remedied in $\PPC$.
The power of diagnostic declines with increasing parameter dimension $p$, as the actual \ac{MMD} value gets closer to the effective support of the null; however, the test still had power to reject the well-specified null even when $p = 50$.

In summary, our additional experiments confirm the intuition that larger dataset sizes $n$ increase our power to detect when the model is misspecified, while larger parameter dimension $p$ decreases our power to detect when the model is misspecified.

\begin{figure}[t!]
    \includegraphics[width=\textwidth]{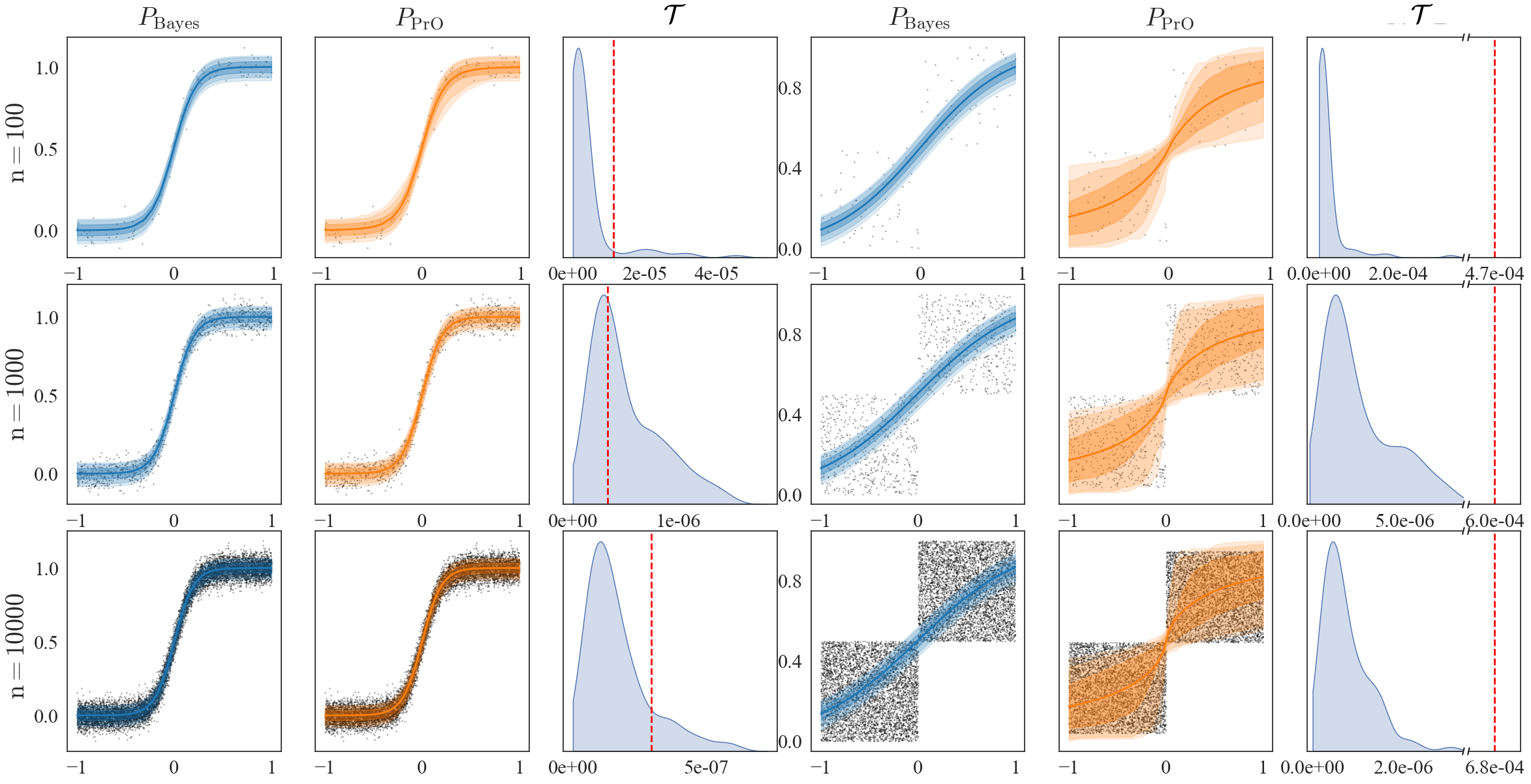}
    \caption{Additional simulation study, varying the size $n$ of the dataset. Each row corresponds to the sigmoid regression task in \Cref{fig: sim study} with different data sizes in which the data are either generated from the statistical model (well-specified, left) or not generated from the statistical model (misspecified, right). The posterior predictive distributions $\PBayes$ and $\PPC$ are displayed, along with the null distribution under the hypothesis that the statistical model is well-specified, and actual realised value of the test statistic $\mathcal{T}$ in \eqref{eq: MMD} (red dashed).
    \alttext{A figure with 18 panels arranged into a $3 \times 6$ grid.  On each row the posterior predictive fits are displayed for both the standard and PrO posteriors, along with the distribution of the test statistic $\mathcal{T}$, when the statistical model is both well-specified and misspecified.  Each row corresponds to a different number of samples $n$ in the dataset.}
    }
    \label{fig: diff data size}
\end{figure}

\begin{figure}[t!]
    \includegraphics[width=\textwidth]{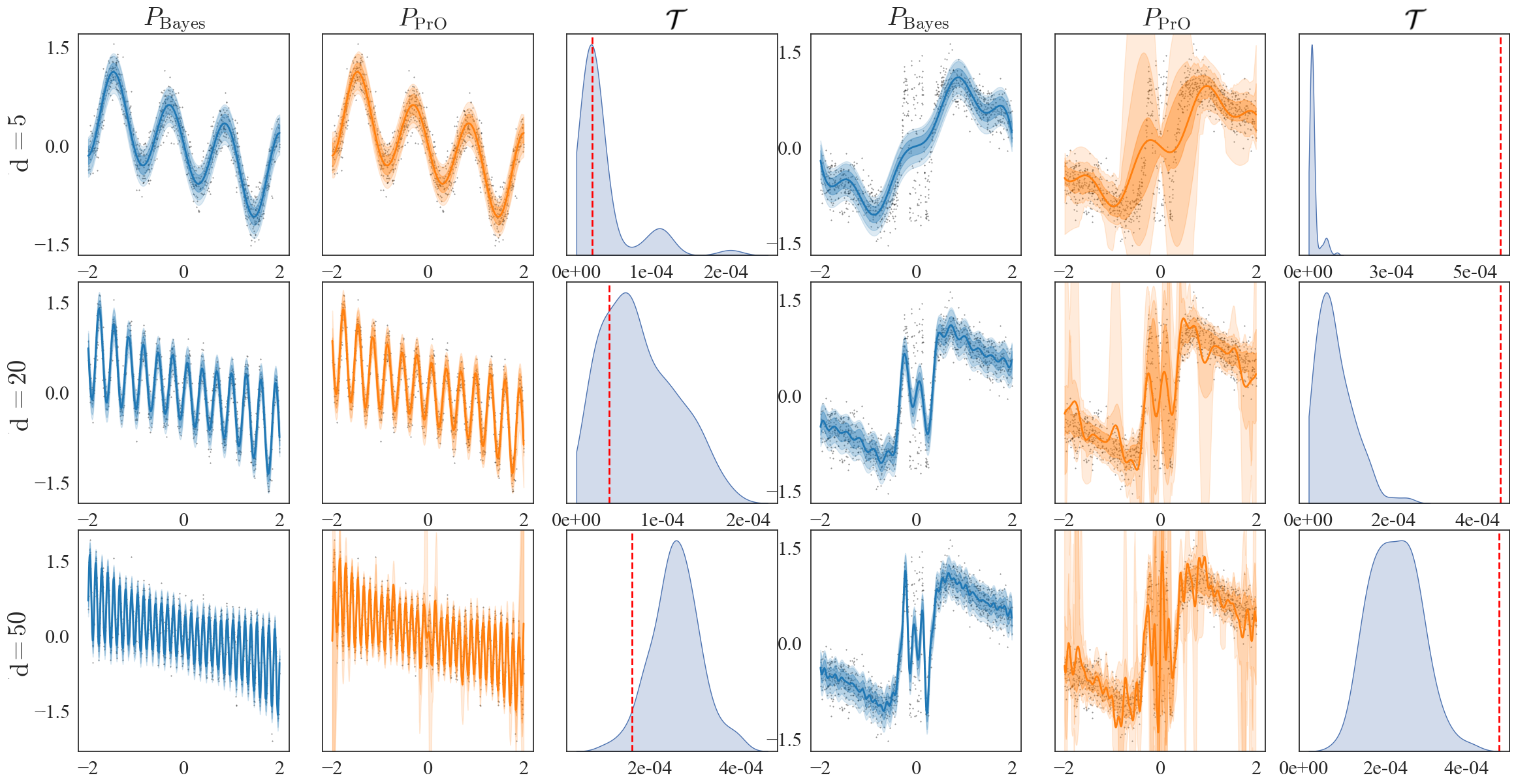}
    \caption{Additional simulation study, varying the number $d$ of parameters in the model. Each row corresponds to a regression task using model \eqref{eq: diff dim model} with different parameters dimension, in which the data are either generated from the statistical model (well-specified, left) or not generated from the statistical model (misspecified, right). The posterior predictive distributions $\PBayes$ and $\PPC$ are displayed, along with the null distribution under the hypothesis that the statistical model is well-specified, and actual realised value of the test statistic $\mathcal{T}$ in \eqref{eq: MMD} (red dashed). 
    \alttext{A figure with 18 panels arranged into a $3 \times 6$ grid.  On each row the posterior predictive fits are displayed for both the standard and PrO posteriors, along with the distribution of the test statistic $\mathcal{T}$, when the statistical model is both well-specified and misspecified.  Each row corresponds to a different number of parameters $d$ in the statistical model.}
    }
    \label{fig: diff dimension}
\end{figure}

\subsection{Insensitivity to the Learning Rate}
\label{app: lr}

This appendix demonstrates the sensitivity analysis described in \Cref{rem: lr sensitivity}.
For each of the three regression tasks in \Cref{app: test problems}, in both the well-specified and the misspecified settings, we computed the parametric bootstrap $p$-value $p_\lambda$ for each $\lambda \in \Lambda = \{2^{-3}, 2^{-2}, 2^{-1}, 1, 2, 2^2, 2^3\}$, using $100$ bootstrap replicates.
The results are displayed in \Cref{fig: lr sensitivity}.
In the well-specified setting $p_\lambda$ is not small for any $\lambda$, while in the misspecified setting the null is rejected for all $\lambda \in \Lambda$. In the quadratic and linear scenarios‌, $p$-values are not sensitive to the learning rate.

\begin{figure}[t!]
\includegraphics[width=\textwidth]{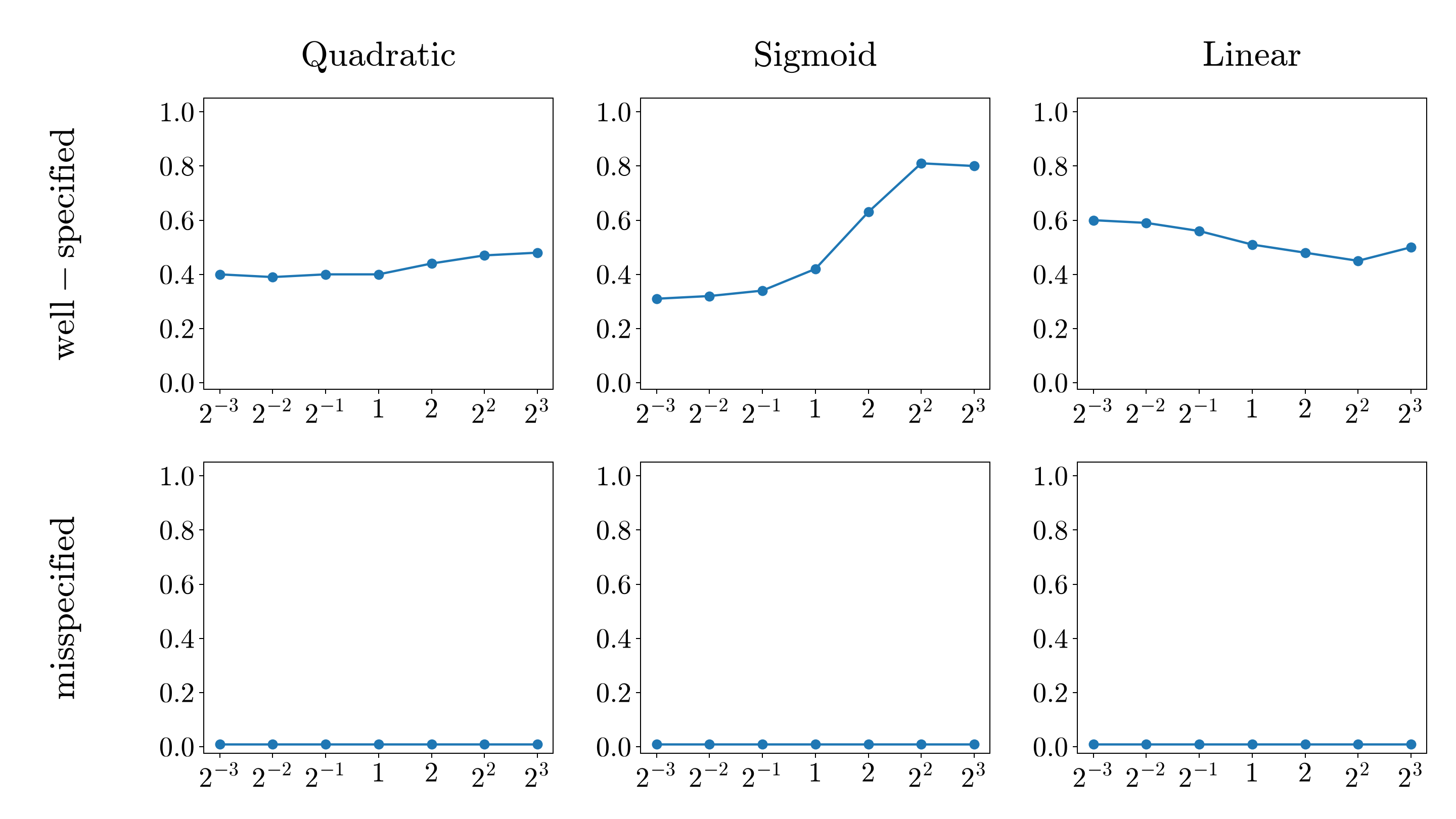}
\caption{ Insensitivity to the learning rate. The parametric bootstrap $p$-value $p_\lambda$ is displayed as a function of the learning rate $\lambda$, for each of the regression tasks in \Cref{fig: sim study}, in the well-specified (top) and misspecified (bottom) settings.}
\label{fig: lr sensitivity}
\end{figure}

\color{black}

\subsection{Details for Seismic Travel Time Tomography}
\label{app: seismic}

For computation using \ac{VGD}, a set of $N = 600$ particles $\{\theta_j^0\}_{j=1}^N$ were initialised by sampling from the prior $Q_0$.
A total of $T = 500$ iterations of \ac{VGD} were performed with step size $\epsilon = 0.1$. 
The Gaussian kernel was used, in line with earlier work in this context, and the length scale was calculated by the median of pairwise distances between particles \citep{garreau2017large}.

\end{document}